\documentclass[a4paper,USenglish,cleveref, autoref]{lipics-v2019}
\pdfoutput=1

\title{On adaptive algorithms for maximum matching}

\author{Falko Hegerfeld}{Humboldt-Universit\"at zu Berlin, Germany}{hegerfeld@informatik.hu-berlin.de}{}{}
\author{Stefan Kratsch}{Humboldt-Universit\"at zu Berlin, Germany}{kratsch@informatik.hu-berlin.de}{}{}

\authorrunning{F. Hegerfeld and S. Kratsch}

\Copyright{Falko Hegerfeld and Stefan Kratsch}



\nolinenumbers 

\hideLIPIcs  

\EventEditors{Christel Baier, Ioannis Chatzigiannakis, Paola Flocchini, and Stefano Leonardi}
\EventNoEds{4}
\EventLongTitle{46th International Colloquium on Automata, Languages, and Programming (ICALP 2019)}
\EventShortTitle{ICALP 2019}
\EventAcronym{ICALP}
\EventYear{2019}
\EventDate{July 9--12, 2019}
\EventLocation{Patras, Greece}
\EventLogo{eatcs}
\SeriesVolume{132}
\ArticleNo{66}

\usepackage{subcaption}
\usepackage{booktabs}
\usepackage{url}

\usepackage[linesnumbered,ruled,vlined]{algorithm2e}

\SetCommentSty{mycommstyle}
\SetKw{True}{true}
\SetKw{False}{false}

\theoremstyle{definition}
\newtheorem{dfn}{Definition}[section]

\theoremstyle{plain}
\newtheorem{thm}[dfn]{Theorem}
\newtheorem{cor}[dfn]{Corollary}
\newtheorem{lem}[dfn]{Lemma}

\crefname{dfn}{Definition}{Definitions}
\crefname{thm}{Theorem}{Theorems}
\crefname{cor}{Corollary}{Corollaries}
\crefname{lem}{Lemma}{Lemmata}
\crefname{prop}{Proposition}{Propositions}
\crefname{rem}{Remark}{Remarks}
\crefname{algorithm}{Algorithm}{Algorithms}
\crefname{algocf}{Algorithm}{Algorithms}

\newcommand{\Oh}{\mathcal{O}} 

\newcommand{\class}{\mathcal{C}} 
\newcommand{\plex}[1]{{\mathcal{P}[#1]}} 
\newcommand{\repl}[1]{{\mathcal{R}[#1]}} 

\newcommand{\dist}{\mathrm{d}} 

\newcommand{\idx}{\mathcal{I}} 

\DeclareMathOperator{\nd}{nd} 
\DeclareMathOperator{\mw}{mw} 
\DeclareMathOperator{\md}{md} 

\newcommand{\NP}{\ensuremath{\textsf{NP}}\xspace}

\begin{document}

\maketitle

\begin{abstract}
In the fundamental \textsc{Maximum Matching} problem the task is to find a maximum cardinality set of pairwise disjoint edges in a given undirected graph. The fastest algorithm for this problem, due to Micali and Vazirani, runs in time $\Oh(\sqrt{n}m)$ and stands unbeaten since 1980. It is complemented by faster, often linear-time, algorithms for various special graph classes. Moreover, there are fast parameterized algorithms, e.g., time $\Oh(km\log n)$ relative to tree-width $k$, which outperform $\Oh(\sqrt{n}m)$ when the parameter is sufficiently small.

We show that the Micali-Vazirani algorithm, and in fact any algorithm following the phase framework of Hopcroft and Karp, is adaptive to beneficial input structure. We exhibit several graph classes for which such algorithms run in linear time $\Oh(n+m)$. More strongly, we show that they run in time $\Oh(\sqrt{k}m)$ for graphs that are $k$ vertex deletions away from any of several such classes, without explicitly computing an optimal or approximate deletion set; before, most such bounds were at least $\Omega(km)$. Thus, any phase-based matching algorithm with linear-time phases obliviously interpolates between linear time for $k=\Oh(1)$ and the worst case of $\Oh(\sqrt{n}m)$ when $k=\Theta(n)$. We complement our findings by proving that the phase framework by itself still allows $\Omega(\sqrt{n})$ phases, and hence time $\Omega(\sqrt{n}m)$, even on paths, cographs, and bipartite chain graphs.
\end{abstract}

\section{Introduction}

The objective in the fundamental \textsc{Maximum Matching} problem is to find a set of disjoint edges of maximum cardinality in a given undirected graph $G=(V,E)$. \textsc{Maximum Matching} has been heavily studied and was the first problem for which a polynomial-time algorithm has explicitly been established \cite{edmonds1965paths}. Several algorithms \cite{blum1990new, gabow1991faster, goldbergkarzanov, micalivazirani} achieve the best known running time of $\Oh(\sqrt{n}m)$ for graphs with $n$ vertices and $m$ edges, starting with the algorithm of Micali and Vazirani~\cite{micalivazirani} in 1980. Since then, the time of $\Oh(\sqrt{n}m)$ remains unbeaten.

This state-of-the-art has motivated extensive research into faster algorithms for \textsc{Maximum Matching} on special inputs: Extensive effort went into beating the worst case running time of $\Oh(\sqrt{n}m)$ for \textsc{Maximum Matching} on special graph classes, resulting in a large number of publications \cite{dragan1997greedy, fouquet1999bipartite, fouquet1997n, gardi2003efficient, glover1967maximum, hopcroftkarp, liang1993finding, madry2013navigating}, often even obtaining linear-time algorithms \cite{chang1996algorithms, dahlhaus1998matching, mertzios2018linear, steiner1996linear, yu1993n}. Similarly, there is a great variety of algorithms whose running time depends on $n$ and $m$ but also on some structural parameter of the input graph, like its tree-width, its genus, or its vertex-deletion distance to a certain graph class (summarized in Table~\ref{tab:known_results}). As an example, one can solve \textsc{Maximum Matching} in time $\Oh(k(n+m))$ when the input graph $G$ is given together with a set $S$ of $k$ vertices such that $G-S$ belongs to a class $\class$ in which \textsc{Maximum Matching} can be solved in linear time (cf.~\cite{mertzios2018linear}): It suffices to solve the problem on $G-S$ in linear time and to then apply at most $k$ augmentation steps to account for vertices in $S$; each such step can be implemented in linear time.

A caveat of this great number of different algorithms for special cases is that we may have to first find the relevant structure, e.g., a set $S$ so that $G-S$ belongs to a certain graph class $\class$, and to then decide which algorithm to apply. In some cases, finding the relevant optimal structure is \NP-hard and using approximate structure may lead to increase in running time. Moreover, except for time $\Oh(\sqrt{k}m)$ relative to vertex cover number $k$ or maximum matching size $k$, which can be seen to follow from the general analysis of Hopcroft and Karp~\cite{hopcroftkarp}, the previously known time bounds improve on $\Oh(\sqrt{n}m)$ only if $k=\Oh(\sqrt{n})$, at best.

\subparagraph{Our results.}
We approach the \textsc{Maximum Matching} problem from the perspective of \emph{adaptive analysis (of algorithms)}. Rather than developing further specialized algorithms for special classes of inputs, we prove that a single algorithm actually achieves the best time bounds relative to several graph classes and parameters; in particular, several new or improved time bounds are obtained. Moreover, that algorithm is known since 1980, namely it is the Micali-Vazirani-algorithm~\cite{micalivazirani}, and it is oblivious to the actual structure and parameter values. In fact, our analysis does not depend on overly specific aspects of that algorithm and, rather, applies to any algorithm for \textsc{Maximum Matching} that follows the ``phase framework'' established by Hopcroft and Karp~\cite{hopcroftkarp} for \textsc{Bipartite Matching} and that implements each phase in linear time, e.g., the algorithms by Blum~\cite{blum1990new} and Goldberg and Karzanov~\cite{goldbergkarzanov}. In this framework, each phase is dedicated to finding a disjoint and maximal packing of shortest augmenting paths, and it can be shown that $\Oh(\sqrt{n})$ phases always suffice (cf.~\cite{hopcroftkarp}).

  \begin{table}
    \centering
    \caption{The second column shows the running times of dedicated algorithms
    for \textsc{Maximum Matching}
    on restricted inputs, which follow as special cases of previous work (see Table~\ref{tab:known_results}). Columns three and four show our results for algorithms employing the phase framework with linear-time phases. Only vertex cover number and matching number had known time $\Oh(\sqrt{k}m)$, others had $\Omega(km)$ prior to our work. The parameters $\mw$ and $\md$ refer to modular-width and modular-depth.}
    \begin{tabular}{lccc}
      \toprule

					& Dedicated algorithm	& \multicolumn{2}{c}{Phase framework algorithms} \\
      Parameter / Graph class			& Parameter $s$ 	& Parameter $s$			& Dist.\ $k$ to parameter $s$ \\
      \cmidrule(r){1-1}
      \cmidrule(lr){2-2}
      \cmidrule(l){3-4}
      Vertex cover number				& $\Oh(\sqrt{s}m)$	& $\Oh(\sqrt{s}m)$	& n.a.\\
      Star forest					& $\Oh(m)$	& $\Oh(m)$			& $\Oh(\sqrt{k}m)$ \\
      Bounded tree-depth				& $\Oh(m)$	& $\Oh(m)$			& $\Oh(\sqrt{k}m)$ \\
      Cluster graph						& $\Oh(m)$	& $\Oh(m)$		& $\Oh(\sqrt{k}m)$ \\
      Minimum degree $n - s$				& $\Oh(sn^2\log n)$	 & $\Oh(sm)$		& $\Oh(\sqrt{k}sm)$ \\
      Independence number				& none	& $\Oh(sm)$				& $\Oh(\sqrt{k}s^2m)$ \\
      Neighborhood diversity				& $\Oh((s^2\log s)n + m)$	& $\Oh(sm)$	 & $\Oh(\sqrt{k}sm)$ \\
      Parameter $\mw$ and $\md$					& $\Oh((\mw^2\log \mw)n + m)$	& $\Oh((c \mw)^{\md}m)$& $\Oh(\sqrt{k}(c \mw)^{\md}m)$\\

      \bottomrule
    \end{tabular}
    \label{tab:positive_results}
  \end{table}

We show that algorithms following the phase framework adapt to beneficial structure in the form of inputs from special graph classes or inputs that are few vertex deletions away from such a class, without running a recognition algorithm or computing the deletion distance (i.e., they are obliviously adaptive). Concretely, we show that any such algorithm solves \textsc{Maximum Matching} in linear time on several graph classes such as cluster graphs or graphs of bounded neighborhood diversity. Moreover, for many such classes we also show that any such algorithm takes time $\Oh(\sqrt{k}m)$ on graphs that are $k$ vertex deletions away from the class, without explicitly computing such a set of deletions (or even knowing the class in question). Furthermore, this running time interpolates between the worst-case time $\Oh(\sqrt{n}m)$ for $k \in \Theta(n)$ and linear time for $k \in \Theta(1)$, hence remaining competitive even in the absence of beneficial input structure. Except for the matching number and the vertex cover number, time bounds of the form $\Oh(\sqrt{k}m)$ are new, even for dedicated algorithms. Besides that, we improve upon the algorithm by Yuster~\cite{yuster2013maximum} for the special case of minimum degree $n - s$ and we improve upon the algorithm by Kratsch and Nelles~\cite{kratschnelles} for the special case of bounded neighborhood diversity. Our positive results are summarized in Table~\ref{tab:positive_results}.

We complement our findings by exhibiting several graph classes on which the phase framework still allows the worst-case of $\Theta(\sqrt{n})$ phases, and hence $\Theta(\sqrt{n}m)$ time. We prove this for paths, trivially perfect graphs, which are a subclass of cographs, and bipartite chain graphs (and hence their superclasses). This, of course, does not contradict the existence of dedicated linear-time algorithms nor the possibility of tweaking a phase-based algorithm to avoid the obstructions. Nevertheless, these results do rule out the possibility of proving adaptiveness of arbitrary phase-based algorithms, and they showcase obstructions that need to be handled to obtain more general adaptive algorithms for \textsc{Maximum Matching}.

\subparagraph{Related work.}
Our work fits into the recent program of ``FPT in P''\footnote{An FPT-algorithm solves a given (usually \NP-hard) problem in time $f(k)n^c$ where $k$ is some problem-specific parameter and $n$ is the input size.} or efficient parameterized algorithms, initiated independently by Abboud et al.~\cite{abboud2016approximation} and Giannopoulou et al.~\cite{giannopoulou2017polynomial}. This program seeks to apply the framework of parameterized complexity to tractable problems to obtain provable running times relative to certain parameters that outperform the fastest known algorithms or (conditional) lower bounds obtained in the fine-grained analysis program (see, e.g.,~\cite{AbboudWY15, Bringmann14, PatrascuW10}). In particular, Mertzios et al.~\cite{mertzios2018linear} have suggested \textsc{Maximum Matching} as the ``drosophila'' of FPT in P, i.e., as a central subject of study, similar to the role that \textsc{Vertex Cover} plays in parameterized complexity. Already, there is a large number of publications on parameterized algorithms for \textsc{Maximum Matching}~\cite{coudertcliquewidth, ducoffepopasplitwidth, ducoffepopapruned, fomin2018fully, kratschnelles, mnn}, apart from large interest in FPT in P in general~\cite{abboud2016approximation, bentert2017parameterized, fluschnik2017can, husfeldtgraphdistances, iwatatreedepth, kellerhals2018parameterized}. There has also been interest in linear-time preprocessing, which, relative to some parameter $k$, reduces the problem to solving an instance of size $f(k)$ and leads to time bounds of the form $\Oh(n+m+g(k))$~\cite{mnnKernel}.

  \begin{table}
    \centering
    \caption{Known parameterized algorithms for \textsc{Maximum Matching}, $k$ denotes the corresponding parameter value and $\omega$ is the matrix multiplication constant.}
    \begin{tabular}[h]{ccc}
      \toprule

      Parameter 		& Running Time 							& Reference \\

      \midrule

      Matching Number 		& $\Oh(\sqrt{k} m)$ 						& \cite{blum1990new, goldbergkarzanov, hopcroftkarp, micalivazirani} \\

      Vertex Cover Number 	& $\Oh(\sqrt{k} m)$ / $\Oh(k(n+m)$ / $\Oh(n + m + k^3)$		& \cite{hopcroftkarp, micalivazirani} / \cite{mnn} / \cite{mnn}\\

      Feedback Vertex Number 	& $\Oh(k(n + m))$ / $\Oh(k n + 2^{\Oh(k)})$ 			& \cite{mnn} / \cite{mnnKernel} \\

      Feedback Edge Number 	& $\Oh(k(n + m))$ / $\Oh(n + m + k^{1.5})$			& \cite{mnn} / \cite{mnnKernel} \\

      $\Delta(G) - \delta(G)$   & $\Oh(k n^2 \log n)$ 						& \cite{yuster2013maximum} \\

      Tree-width 		& $\Oh(k^4 n\log n)$ / $\Oh(k m \log n)$ 			& \cite{fomin2018fully} / \cite{iwatatreedepth} \\

      Tree-depth 		& $\Oh(k m)$ 							& \cite{iwatatreedepth} \\

      Modular-width 		& $\Oh(k^4 n + m)$ / $\Oh((k^2 \log k) n + m)$			& \cite{coudertcliquewidth} / \cite{kratschnelles} \\

      Split-width 		& $\Oh((k \log^2 k)(n + m)\log n)$ 				& \cite{ducoffepopasplitwidth} \\

      $P_4$-sparseness 		& $\Oh(k^4(n + m))$ 						& \cite{coudertcliquewidth} \\

      Genus 			& $\Oh(f(k)n^{\omega/2})$ 					& \cite{yuster2007minors} \\

      $H$-minor-free 		& $\Oh(f(H)n^{3\omega/(\omega + 3)})$ 				& \cite{yuster2007minors} \\

      Dist.\ to Cocomparability
				& $\Oh(k (n + m))$ 						& \cite{mertzios2018linear} \\

      Distance to Chain Graph   & $\Oh(k(n + m))$ / $\Oh(n + m + k^3)$				& \cite{mnn} / \cite{mnnKernel} \\

      \bottomrule
    \end{tabular}
    \label{tab:known_results}
  \end{table}

Adaptive analysis of algorithms has been most successful in the context of sorting and searching~\cite{BarbayFN12,BarbayN13,DisserK17,Estivill-CastroW92}. We are not aware of prior (oblivious) adaptive analysis of established algorithms for the \textsc{Maximum Matching} problem but two works have designed dedicated adaptive algorithms relative to tree-depth~\cite{iwatatreedepth} and modular-width~\cite{kratschnelles}.

Bast et al.~\cite{bast2006matching} analyzed the Micali-Vazirani algorithm for random graphs, obtaining a running time of $\Oh(m \log n)$ with high probability.

\subparagraph{Organization.}
Some preliminaries on graphs and matchings are recalled in \cref{section:preliminaries}. \cref{section:hopcroftkarp} is dedicated to recalling the analysis of Hopcroft and Karp~\cite{hopcroftkarp} and defining phase-based algorithms. In \cref{section:upperbounds} we present the positive results and the lower bounds are presented in Section~\ref{section:lowerbounds}. We conclude in \cref{section:conclusion}.

\section{Preliminaries}\label{section:preliminaries}

  We mostly consider simple graphs, unless stated otherwise, and denote an edge between $v$ and $w$ as the concatenation of its endpoints, $vw$.
  Let $G = (V, E)$ denote a graph.
  A path $P$ in $G$ is denoted by listing its vertices in order, i.e., $P = v_1 v_2 \ldots v_\ell$. We use the following notation to refer to subpaths of a path $P$:
  \begin{align*}
    P_{[v_i, v_j]} = \begin{cases} v_i v_{i+1} \ldots v_j & \text{if } i \leq j, \\ v_i v_{i-1} \ldots v_j & \text{if } i > j. \end{cases}
  \end{align*}
  By \emph{disjoint} paths we will always mean \emph{vertex-disjoint} paths.
  For a set of vertices $S \subseteq V$, we define $\delta(S) = \left\{vw \in E \, : \, v \in S, w \notin S \right\}$. For two sets $X, Y$ their
  \emph{symmetric difference} is denoted by $X \triangle Y = (X \setminus Y) \cup (Y \setminus X)$.

  A set $C \subseteq V$ is a \emph{vertex cover} of $G$ if every edge of $G$ has at least one endpoint in $C$; the \emph{vertex cover number} $\tau(G)$ of $G$
  is the minimum cardinality of any vertex cover of $G$. An \emph{independent set} of $G$ is a set $S \subseteq V$ of pairwise nonadjacent vertices;
  the \emph{independence number} $\alpha(G)$ of a graph $G$ is the maximum cardinality of any independent set of $G$. The \emph{maximum degree} and \emph{minimum degree} of $G$
  are denoted by $\Delta(G)$ and $\delta(G)$ respectively.
  For a class $\class$ of graphs we define $\dist_{\class}(G) = \min_{S \subseteq V : G - S \in \class} |S|$ to be the \emph{vertex deletion distance}
  of $G$ to $\class$; a set $S$ such that $G-S\in\class$ is called a \emph{modulator}. E.g., the vertex cover number $\tau(G)$ is the vertex deletion distance of $G$
  to edgeless graphs, i.e., independent~sets.

  A \emph{matching} in a graph is a set of pairwise disjoint edges. Let $G = (V, E)$ be a graph and let $M \subseteq E$ be a matching in $G$.
  The matching $M$ is \emph{maximal} if there is no matching $M'$ in $G$ such that $M \subsetneq M'$ and $M$ is \emph{maximum} if there is no matching $M'$ in $G$
  such that $|M| < |M'|$;
  the \emph{matching number} $\nu(G)$ of $G$ is the cardinality of a maximum matching in $G$.

  A vertex $v \in V$ is called $M$-\emph{matched} if there is an edge in $M$ that contains $v$, and $v$ is called $M$-\emph{exposed} otherwise.
  We do not mention $M$ if the matching is clear from the context.
  We say that an edge $e \in E$ is \emph{blue} if $e \notin M$ and $e$ is \emph{red} if $e \in M$. An $M$-\emph{alternating path} is a path in $G$ that
  alternatingly uses red and blue edges. An $M$-\emph{augmenting path} is an $M$-alternating path that starts and ends with an $M$-exposed vertex;
  a \emph{shortest} $M$-augmenting path is an $M$-augmenting path that uses as few edges as possible.

  It is well known that matchings can be enlarged along augmenting paths and that an augmenting path always exists if the matching is not maximum. We say that the matching $M \triangle E(P)$ is obtained by \emph{augmenting} $M$ along $P$.

  \begin{lem}
    If $M$ is a matching in $G$ and $P$ is an $M$-augmenting path, then $M\triangle E(P)$ is also a matching in $G$ and has size $|M \triangle E(P)| = |M| + 1$.
  \end{lem}

  \begin{thm}[\cite{hopcroftkarp}]
    \label{existence_disjoint_aug_paths}
    Let $M$ and $N$ be matchings in $G=(V,E)$ with $|N| > |M|$. The sub\-graph $G' = (V, M \triangle N)$ of $G$ contains at least $|N|-|M|$ vertex-disjoint
    $M$-augmenting~paths.
  \end{thm}

  \begin{cor}[\cite{berge1957two}]
    A matching $M$ is maximum if and only if there is no $M$-augmenting path.
  \end{cor}


\section{Hopcroft-Karp analysis}\label{section:hopcroftkarp}
  Many of the fastest algorithms for \textsc{Maximum Matching} make use of a framework introduced by Hopcroft and Karp~\cite{hopcroftkarp} for the special case of bipartite matching. We give an
  overview of the framework in this section, mostly following Hopcroft and Karp~\cite{hopcroftkarp}.
  The main idea is to search for shortest augmenting paths instead of arbitrary augmenting paths. Exhaustively searching for shortest augmenting paths and augmenting along them leads to \cref{generic_matching_algo}.

  \begin{center}
    \begin{algorithm}[H]
      \KwIn{Graph $G$}
      \KwOut{Maximum matching $M$}
      $M \leftarrow \emptyset$\;
      \While{$M$ is not maximum}
      {
	Find a shortest $M$-augmenting path $P$\;
	$M \leftarrow M \triangle E(P)$\;
      }
      \Return $M$\;
      \caption{Generic Matching Algorithm}
      \label{generic_matching_algo}
    \end{algorithm}
  \end{center}

  Let $P_1, P_2, \ldots, P_\ell$ be the sequence of augmenting paths in the order found during an execution of \cref{generic_matching_algo}. Hopcroft and Karp~\cite{hopcroftkarp} observed the following properties of the computed shortest augmenting paths.

  \begin{lem}[\cite{hopcroftkarp}]
    Let $M$ be a matching, let $P$ be a shortest $M$-augmenting path, and let $P'$ be a $M \triangle E(P)$-augmenting path,
    then $|P'| \geq |P| + 2|P \cap P'|$.
  \end{lem}

  \begin{cor}[\cite{hopcroftkarp}]
    \label{path_length_increasing}
    The sequence $|P_1|, \ldots, |P_\ell|$ is non-decreasing.
  \end{cor}

  \begin{cor}[\cite{hopcroftkarp}]
    \label{same_length_disjoint}
    If $|P_i| = |P_j|$ for some $i \neq j$, then $P_i$ and $P_j$ are vertex-disjoint.
  \end{cor}

  Following these observations, we can partition the sequence $P_1, \ldots, P_\ell$ into maximal contiguous subsequences $P_i, P_{i + 1}, \ldots, P_j$
  such that $|P_i| = |P_{i + 1}| = \cdots = |P_j| < |P_{j+1}|$, due to
  \cref{same_length_disjoint} the paths $P_i, P_{i + 1}, \ldots, P_j$ must be pairwise vertex-disjoint.
  Every such subsequence is called a \emph{phase} and corresponds to a maximal set of vertex-disjoint shortest augmenting paths due to \cref{same_length_disjoint}.
  With the terminology of phases introduced, it is useful to restate \cref{generic_matching_algo} as follows.

  \begin{center}
    \begin{algorithm}[H]
      \KwIn{Graph $G$}
      \KwOut{Maximum matching $M$}
      $M \leftarrow \emptyset$\;
      \While{$M$ is not maximum}
      {
	Find a maximal set $S$ of vertex-disjoint shortest $M$-augmenting paths\;
	Augment $M$ along all paths in $S$\;
      }
      \Return $M$\;
      \caption{Phase Framework}
      \label{phase_framework}
    \end{algorithm}
  \end{center}

  In \cref{phase_framework} each iteration of the \textbf{while}-loop corresponds to a single phase.
  If an algorithm implements \cref{phase_framework} we say that it \emph{employs the phase framework}. In the following, we will abstract from the implementation details of
  algorithms employing the phase framework and only bound the number of phases that are required in the worst case. Hopcroft and Karp~\cite{hopcroftkarp} presented an upper bound in terms of the matching number $\nu(G)$.

  \begin{thm}[\cite{hopcroftkarp}]
    \label{matching_number_bound}
    Every algorithm employing the phase framework requires at most \newline$2 \left\lceil \sqrt{\nu(G)} \right\rceil + 2$ phases.
  \end{thm}

  The next bound is a simple corollary of \cref{matching_number_bound} by noticing that $\nu(G) \leq \frac{n}{2}$, but we opt
  to give an independent proof to serve as an instructive example for the proofs to come.

  \begin{thm}[folklore]
    \label{classical_bound}
    Every algorithm employing the phase framework requires at most $\Oh(\sqrt{n})$ phases.
  \end{thm}

  \begin{proof}
    Let $M$ denote the matching obtained after performing $\left\lceil \sqrt{n} \right\rceil$ phases of \cref{phase_framework}. Every further $M$-augmenting path has length at least $\left \lceil \sqrt{n} \right\rceil$ by \cref{path_length_increasing}.
    This implies that we can pack at most $\left\lceil \sqrt{n} \right\rceil$ such augmenting paths into $G$ and hence by \cref{existence_disjoint_aug_paths}, at most $\left\lceil \sqrt{n} \right\rceil$
    augmentations remain. Since we perform at least one augmentation per phase, \cref{phase_framework} must have terminated after an additional $\left\lceil \sqrt{n} \right\rceil$ phases.
    Thus, \cref{phase_framework} terminates after at most $2 \left\lceil \sqrt{n} \right\rceil$ phases.
  \end{proof}

  Several of the fastest \textsc{Maximum Matching} algorithms employ the phase framework \cite{blum1990new, goldbergkarzanov, micalivazirani}.
  Any one of these algorithms yields the following time bound for a single phase.

  \begin{thm}
    \label{linear_phase}
    There is an algorithm that given a matching $M$ computes a maximal set of vertex-disjoint shortest $M$-augmenting paths in time $\Oh(m)$. In particular, each phase
    of the phase framework can be implemented to run in time $\Oh(m)$.
  \end{thm}

  With \cref{classical_bound} and \cref{matching_number_bound} we obtain the following time bounds.

  \begin{thm}
    \label{matching_number_runtime}
    There is an algorithm employing the phase framework that solves \textsc{Maximum Matching} in time $\Oh(\sqrt{n} m)$ and $\Oh(\sqrt{\nu(G)} m)$.
  \end{thm}


\section{Adaptive parameterized analysis}\label{section:upperbounds}

  In this section we will perform an adaptive analysis for algorithms employing the phase framework by analyzing the required number of phases
  in terms of various graph parameters. \cref{linear_phase} yields an improved running time if the considered parameter is small enough.

  Many of the considered parameters are NP-hard to compute, e.g., the vertex cover number. This is not an issue, however, as we only require the parameter value
  for the running time analysis and not for the execution of the algorithm. In this sense, algorithms employing the phase framework are proved to
  obliviously adapt to the studied parameters.


\subsection{Short alternating paths}

  Our main lemma relies on bounding the length of shortest augmenting paths. In the interest of simplifying later arguments, we will not only bound the length of shortest augmenting paths,
  but also of alternating paths that are not necessarily augmenting. The strategy for obtaining such upper bounds is to take a long alternating path $P$ and deduce that additional edges must exist in
  $G[V(P)]$ that enable us to find a shorter alternating path $P'$ in $G[V(P)]$ between the endpoints of $P$. Hence, the replacement path $P'$ will only visit vertices
  that are also visited by the original path $P$. The following definition formalizes this idea.

  \begin{dfn}
    Given a matching $M$ in a graph $G$ and an $M$-alternating path $P = v_1 \ldots v_\ell$. We say that an $M$-alternating path
    $P' = w_1 \ldots w_k$ \emph{replaces} $P$ if the following is true:
    \begin{itemize}
    \item $V(P') \subseteq V(P)$,
    \item $w_1 = v_1$ and $w_k = v_\ell$,
    \item $P'$ has the same parity as $P$ with respect to $M$, i.e., \newline $v_1v_2 \in M \iff w_1w_2 \in M$ and $v_{\ell - 1}v_\ell \in M \iff w_{k-1}w_k \in M$.
    \end{itemize}
    In particular, if $P'$ replaces $P$, then $P'$ is at most as long as $P$.
  \end{dfn}

  A technicality that arises from considering general alternating paths, as opposed to augmenting paths, is that an alternating path that starts and ends with a blue edge, i.e. an edge not in the matching, might have endpoints
  that are not exposed. If we want to shortcut by taking a different edge incident to such an endpoint, then this edge might be red which causes our constructions
  to fail. To avoid this issue, it suffices to consider the subpath resulting from the removal of the first and last edge.

  \begin{dfn}
    A graph $G$ is $\ell$-\emph{replaceable} if for every matching $M$ each $M$-alternating path can be replaced by an $M$-alternating path of length at most $\ell$.
    The class of $\ell$-replaceable graphs is denoted $\repl{\ell}$.
  \end{dfn}

  We will now show that algorithms employing the phase framework require only few phases
  for graphs that are close, in the sense of vertex deletion distance, to $\ell$-replaceable graphs.

  \begin{lem}
    \label{replaceable_phases}
    Every algorithm employing the phase framework requires at most $\Oh(\sqrt{k} \ell)$ phases on graphs $G$ with $\dist_{\repl{\ell}}(G) \leq k$.
    In particular, \textsc{Maximum Matching} can be solved in time $\Oh(\sqrt{k} \ell m)$ for such graphs.
  \end{lem}

  \begin{proof}
    Let $S \subseteq V$ with $|S| \leq k$, such that $G - S \in \repl{\ell}$ and let $M$ be the matching obtained after performing $\lceil \sqrt{k} \ell \rceil$ phases of \cref{phase_framework}.
    We claim that every shortest $M$-augmenting path uses at least $\lfloor \sqrt{k} \rfloor$ vertices of $S$; every such path has a length of at least $\lceil \sqrt{k} \ell \rceil$ due to \cref{path_length_increasing}. Consider such a path $P$,
    since $G - S$ is $\ell$-replaceable, we can assume that every time $P$ enters $G - S$ it uses at most $\ell + 1$ vertices of $G - S$ before going back to $S$.
    Hence, $P$ must use at least $\lfloor \sqrt{k} \rfloor - 1$ vertices of $S$ to have a length of $\lceil \sqrt{k}\ell \rceil$ or more.

    Since $S$ is of size at most $k$ and due to the properties of replacing paths, this implies that we can pack at most $2 \lceil \sqrt{k} \rceil$ $M$-augmenting paths into $G$ as
    \begin{align*}
    2 \left\lceil \sqrt{k} \right\rceil \left(\left\lfloor \sqrt{k} \right\rfloor - 1\right) \geq 2 \sqrt{k} \left(\sqrt{k} - 2\right) = 2k - 4\sqrt{k} \geq k \quad \text{ for } k \geq 16.
    \end{align*}
    By \cref{existence_disjoint_aug_paths} at most $2 \lceil \sqrt{k} \rceil$ augmentations remain, which require at most $2 \lceil \sqrt{k} \rceil$ phases.
    In total, we need $\lceil \sqrt{k} \ell \rceil + 2 \lceil \sqrt{k} \rceil \in \Oh(\sqrt{k} \ell)$ phases. \cref{linear_phase} implies the time bound. \qedhere
  \end{proof}

  The following running time bound relative to the vertex cover number $\tau(G)$ follows directly from \cref{matching_number_runtime} by the use of
  the well-known inequality $\nu(G) \leq \tau(G)$.

  \begin{thm}
    \label{vertex_cover}
    Every algorithm employing the phase framework requires at most $\Oh(\sqrt{\tau(G)})$ phases. In particular,
    \textsc{Maximum Matching} can be solved in time $\Oh(\sqrt{\tau(G)}m)$.
  \end{thm}

  Alternatively, observe that if $\class$ is the class of independent sets, then $\dist_{\class}(G) = \tau(G)$ and
  hence \cref{vertex_cover} is implied by \cref{replaceable_phases} as independent sets are trivially $1$-replaceable.

  More generally, every graph without paths of length $\ell + 1$ is $\ell$-replaceable. It is known that a graph class has bounded path length if and only if
  it has bounded tree-depth (see, e.g., \cite[Chapter 6]{nevsetril2012sparsity}).
  As a further special case consider the class of \emph{star forests}, i.e., graphs where every connected component is a star and therefore must be $2$-replaceable.
  Hence, we obtain the following corollary of \cref{replaceable_phases}.

  \begin{cor}
    Let $\class$ be the class of star forests. Every algorithm employing the phase framework requires at most $\Oh(\sqrt{\dist_\class(G)})$ phases. In particular,
    \textsc{Maximum Matching} can be solved in time $\Oh(\sqrt{\dist_\class(G)} m)$.
  \end{cor}


\subsection{Independence number}

  \begin{figure}
    \centering
    \includegraphics[scale = 0.94]{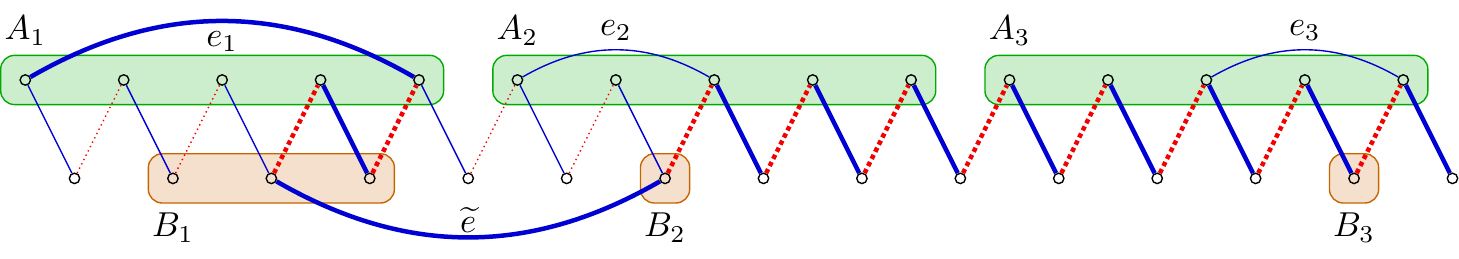}
    \caption{The construction in \cref{ind_num_replaceable} for $\alpha(G) = 2$, the red dotted edges are matched and the blue edges are unmatched. The thickened edges represent the shorter alternating path $P'$.}
    \label{fig:ind_num_proof}
  \end{figure}

  A graph with independence number $k$ contains many edges in the sense that any set of $k+1$ vertices
  must induce an edge. We will use this property to shorten long alternating paths.

  \begin{thm}
    \label{ind_num_replaceable}
    Suppose that $G$ is a graph such that $\alpha(G) \leq k$, then $G$ is $\Oh(k^2)$-replaceable.
  \end{thm}

  \begin{proof}
    First, fix a matching $M$. We show how to replace alternating paths that begin and end with a blue edge, i.e., an edge not in $M$. By replacing appropriate subpaths
    of long alternating paths with other parities, the general result will follow.

    Suppose that $P=a_1 b_1 a_2 b_2 \ldots a_\ell b_\ell$ is an alternating path that is longer than $4(k+1)^2 + 1$ and that
    starts and ends with a blue edge. We can assume that $\ell$ is odd.
    Distinguishing the vertices of $P$ by their parity, we define $A = \{a_i \, : \, i \in [\ell] \}$ and $B = \{b_i \, : \, i \in [\ell] \}$.
    Furthermore, we define the sets $A_i$ and $A'_i$ for $i = 1, \ldots, k + 1$ by
    \begin{align*}
      A_i = \{a_{(i-1)(2k+2)+j} \, : \, j = 1, 2, 3, \ldots, 2k + 1\} \text{ and } A'_i = \{a_j \in A_i \, : \, j \text{ odd}\}.
    \end{align*}
    Note that $|A'_i| = k + 1 > k$ for all $i$, hence the $A'_i$ cannot be independent sets.
    Thus, there is at least one edge $e_i$ in $G[A'_i]$.
    We denote the endpoints of $e_i$ by
    \begin{align}
    \label{eq:ind_num_first}
      e_i = a_{p_i}a_{q_i}, \mbox{ where } p_i \leq q_i - 2, \text{ for all } i = 1, \ldots, k + 1.
    \end{align}

    We now consider the vertices of $B$ that lie between the endpoints of $e_i$ on $P$; we omit $b_{p_i}$ to ensure that the constructed path is shorter than $P$. Concretely, let $B_i= \{b_{p_i + 1}, b_{p_i + 2}, \ldots, b_{q_i - 1} \} \text{ for all } i = 1, \ldots, k + 1.$
    \cref{eq:ind_num_first} implies that $B_i \neq \emptyset$.
    Now, we arbitrarily choose a vertex $b_i'$ from each $B_i$ and define $B' = \{b_i' \, : \, i = 1, \ldots, k + 1\}$.
    Observe that $B'$ cannot be an independent set as $B'$ contains $k + 1$ vertices; thus, there must exist an edge $\tilde{e} = b_i'b_j'$ with $i < j$.
    We construct the path $P'$ that replaces $P$ by (see \cref{fig:ind_num_proof})
    \begin{align*}
      P' = P_{[a_1, a_{p_i}]} P_{[a_{q_i}, b_i']} P_{[b_j', b_\ell]},
    \end{align*}
    using edges $a_{p_i}a_{q_i}$ and $b_i'b_j'$; note that $P_{[a_{q_i}, b_i']}$ is a subpath of $P$ in reverse order. Note that both edges are blue because $a_{q_i}$ and $b_i'$ are already incident with red edges on $P$.
    It can be easily checked that $P'$ is an alternating path and, in particular, that it is a valid replacement for $P$. We will now show that $P'$ is strictly shorter than $P$.
    It suffices to compare the length of $P_1 = P_{[a_{p_i}, b_j']}$ and $P_2 = P_{[a_{p_i}, b_j']}' =  a_{p_i}P_{[a_{q_i}, b_i']}b_j'$
    as $P$ and $P'$ agree on the remaining parts. Let $s$ and $t$ be the indices such that $b_i' = b_s$ and $b_j' = b_t$.
    The length of $P_1$ is $2(t-p_i) + 1$. For the length of $P_2$ we obtain
    \begin{align*}
      |P_2|	& = 1 + (2(s-q_i)+1) + 1 = 2(s-q_i) + 3 < 2(s-(p_i + 1)) + 3  = 2(s-p_i) + 1 \\
		& < 2(t-p_i) + 1 = |P_1|,
    \end{align*}
    where the first inequality follows from \cref{eq:ind_num_first}.
  \end{proof}

  By combining this result with \cref{replaceable_phases} we obtain the following corollary.

  \begin{cor}
    Every algorithm employing the phase framework requires at most $\Oh(\alpha(G)^2)$ phases.
    In particular, \textsc{Maximum Matching} can be solved in time $\Oh(\alpha(G)^2 m)$.

    Let $\class_k$ denote the class of graphs with independence number at most $k$ in each connected component. Every algorithm employing the phase framework
    requires at most $\Oh(\sqrt{\dist_{\class_k}(G)} k^2)$ phases. In particular, \textsc{Maximum Matching} can be solved in time $\Oh(\sqrt{\dist_{\class_k}(G)} k^2 m)$.
  \end{cor}

  A better analysis in terms of the independence number can be achieved by not using replaceability, but in exchange we lose the square root dependence
  on the size of the modulator. 

  \begin{lem}
    \label{ind_num_maximal}
    Every maximal matching $M$ covers at least $n - \alpha(G)$ vertices.
  \end{lem}

  \begin{proof}
    If $\alpha(G) + 1$ vertices were exposed, they would not be an independent set and hence have an edge between them. Thus, $M$ would not be maximal.
  \end{proof}

  \begin{cor}
    Every algorithm employing the phase framework requires at most $\Oh(\alpha(G))$ phases.
    In particular, \textsc{Maximum Matching} can be solved in time $\Oh(\alpha(G) m)$.
  \end{cor}

  \begin{proof}
    After the first phase, we know that at most $\alpha(G)$ vertices are exposed by \cref{ind_num_maximal}.
    Hence, $\alpha(G)/2$ further augmentations suffice.
  \end{proof}


\subsection{\texorpdfstring{$s$-plexes}{s-plexes}}

  An $s$-\emph{plex}, $s \geq 1$, is an $n$-vertex graph $G$ with minimum degree $\delta(G) \geq n-s$. They were introduced by Seidman and Foster~\cite{seidman1978graph}
  as a generalization of cliques, which are $1$-plexes. Problems related to $s$-plexes have been studied in parameterized complexity before \cite{guo2009more, van2012approximation}.
  Let $\plex{s}$ denote the class of graphs that are disjoint unions of $s$-plexes. Given a vertex $v$ and a set $W \subseteq V \setminus \{v\}$ of size at
  least $s$ in an $s$-plex, we know that there is an edge $vw$ for some $w \in W$. Thus, $s$-plexes allow for better control than a
  small independence number as we can guarantee the existence of an edge incident to some specific vertex instead of just getting some edge in
  a large set of vertices. 

  \begin{thm}
    \label{s_plex_replaceable}
    If $G$ is a $k$-plex, i.e., if $\delta(G) \geq n - k$, then $G$ is $\Oh(k)$-replaceable.
  \end{thm}

  \begin{proof}
    Let $M$ be a matching in $G$ and suppose that $P = v_1 \ldots v_\ell$ is an $M$-alternating path in $G$ of length
    $\ell - 1 \geq 2k + 5$. Let $i$ be the smallest integer such that $v_i$ is $M$-exposed or $v_{i - 1} v_i$ is red, i.e., $v_{i - 1} v_i\in M$.
    Let $j$ be the largest integer such that $v_j$ is $M$-exposed or $v_j v_{j+1}$ is red.
    We have $i \in \{1, 2, 3\}$ and $j \in \{\ell - 2, \ell - 1, \ell\}$.

    Consider the vertices $W = \{v_j, v_{j-2}, v_{j-4}, \ldots, v_{j-2(k-1)}\}$ and observe that $v_i \notin W$ as $3 < \ell - 2k$.
    The set $W$ contains $k$ vertices and since $G$ is a $k$-plex there must be some $v_s \in W$ such that $v_iv_s \in E$.
    By choice of $i$ and $s$ the edges $v_i v_{i+1}$ and $v_{s-1} v_s$ must be blue. Similarly, by choice of $i$, the edge $v_iv_s$ is blue. Hence, as seen in \cref{fig:s_plex_proof}, we can replace $P$ by the shorter $M$-alternating path
    $P' = P_{[v_1, v_i]} P_{[v_s, v_\ell]}$ with length $|P'| \leq 2 + 2(k-1) + 2 = 2k + 2$.
  \end{proof}
  
  \begin{figure}
    \centering
    \includegraphics[scale = 1]{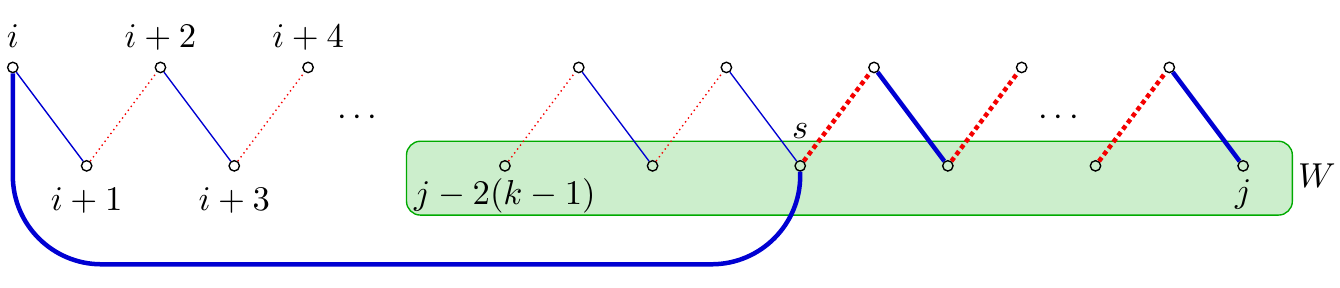}
    \caption{Replacing alternating paths in a $k$-plex.}
    \label{fig:s_plex_proof}
  \end{figure}

  By \cref{replaceable_phases}, we obtain the following corollary.

  \begin{cor}
    Every algorithm employing the phase framework requires at most $\Oh(n - \delta(G))$ phases.
    In particular, \textsc{Maximum Matching} can be solved in time $\Oh((n - \delta(G))m)$.

    Every algorithm employing the phase framework requires at most $\Oh(\sqrt{\dist_{\plex{s}}(G)} s)$ phases.
    Hence, \textsc{Maximum Matching} can be solved in time $\Oh(\sqrt{\dist_{\plex{s}}(G)} sm)$.
    In particular, \textsc{Maximum Matching} can be solved in time $\Oh(\sqrt{k}m)$ on graphs that
    have distance at most $k$ to cluster graphs.
  \end{cor}

  
\subsection{Neighborhood diversity}
\label{sec_nbhd_div}
  In this section we consider the parameter neighborhood diversity, which was introduced by Lampis~\cite{lampis2012algorithmic} and has seen further investigation in
  parameterized complexity research \cite{fialadistancelabeling, ganian2012using, ganian2013expanding, knop2017simplified}.

  \begin{dfn}[\cite{lampis2012algorithmic}]
    Two vertices $v$ and $w$ have the same type if $N(v) \setminus \{w\} = N(w) \setminus \{v\}$.
    This defines an equivalence relation on the vertices of $G$. The \emph{neighborhood diversity} of $G$,
    denoted by $\nd(G)$, is the number of equivalence classes of this equivalence relation.
  \end{dfn}

  \begin{thm}
    \label{neighborhood_diversity_replaceable}
    If $G$ is a graph such that $\nd(G) \leq k$, then $G$ is $\Oh(k)$-replaceable.
  \end{thm}

  \begin{proof}
    Again, we will only consider alternating paths starting and ending with a blue edge.
    Let $M$ be a matching in $G$ and suppose that $P = v_1 \ldots v_\ell$ is an $M$-alternating path
    starting and ending with a blue edge and of length $\ell - 1 \geq 6k + 3$.
    Consider the set $V' = \{v_2, v_5, \ldots, v_{6k + 2} \}$, which contains $2k+1$ vertices.
    By the pigeonhole principle, there must be three vertices $v_r, v_s, v_t \in V'$, with $r < s < t$,
    that are of the same type. Define $u_i = v_{i-1}$ and $w_i = v_{i+1}$ for $i \in \{r,s,t\}$.

    Now, consider the three edges $u_r v_r, u_s v_s$ and $u_t v_t$. We distinguish between two cases and
    will construct a shorter alternating path $P'$ that replaces $P$ in either case.

    \begin{enumerate}
    \item At least two of the edges $u_r v_r, u_s v_s, u_t v_t$ are red.

    Without loss of generality assume that $u_r v_r$ and $u_s v_s$ are red, then
    \begin{align*}
      P' = P_{[v_1, v_r]} P_{[w_s, v_\ell]}
    \end{align*}
    is shorter than $P$ and replaces $P$. The edge $v_r w_s$ must exist because $v_r$ and $v_s$ have the same type,
    i.e., $w_s \in N(v_s) \setminus \{v_r\} = N(v_r) \setminus \{v_s\}$ as clearly $v_r \neq w_s$ and $v_s \neq w_s$.

    \item At least two of the edges $u_r v_r, u_s v_s, u_t v_t$ are blue.

    Without loss of generality assume that $u_r v_r$ and $u_s v_s$ are blue, then
    \begin{align*}
      P' = P_{[v_1, u_r]} P_{[v_s, v_\ell]}
    \end{align*}
    is shorter than $P$ and replaces $P$. The edge $u_r v_s$ must exist because $v_r$ and $v_s$ have the same type.
    \end{enumerate}
    Hence, any $M$-alternating path can be replaced by a path of length at most $\Oh(k)$.
  \end{proof}

  Combining the replaceability result, \cref{neighborhood_diversity_replaceable}, with the main lemma, \cref{replaceable_phases},
  yields the following corollary.

  \begin{cor}
    Every algorithm employing the phase framework requires at most $\Oh(\nd(G))$ phases.
    In particular, \textsc{Maximum Matching} can be solved in time $\Oh(\nd(G) m)$.

    Let $\class_k$ denote the class of graphs such that each connected component has neighborhood diversity at most $k$. Every algorithm
    employing the phase framework requires at most $\Oh(\sqrt{d_{\class_k}(G)} k)$ phases. In particular, \textsc{Maximum Matching} can be solved
    in time $\Oh(\sqrt{d_{\class_k}(G)} k m)$.
  \end{cor}


\subsection{Modular decomposition}
\label{sec_mod_dec}
  The concept of \emph{modular decompositions} was introduced by Gallai~\cite{gallai1967transitiv} to aid in the recognition of comparability graphs. Several linear time algorithms have been given to
  compute the modular decomposition of a graph \cite{cournier1994new, mcconnell1999modular, tedder2008simpler}. Several articles have studied problems parameterized by a width measure related
  to the modular decomposition, called \emph{modular-width} \cite{coudertcliquewidth, Fomin2018, kratschnelles}. On a high level, the modular decomposition recursively partitions the vertex set
  of a graph into parts that have simple interactions between each other. We will now give the formal definitions.

  Let $G = (V, E)$ be a graph. A vertex set $M \subseteq V$ is a \emph{module} if for all $v, w \in M$ it holds that $N(v) \cap \overline{M} = N(w) \cap \overline{M}$. The modules $\emptyset$, $V$, and singletons are called \emph{trivial}. A module $M$ is \emph{strong} if for every other module $M'$ of $G$ we have that $M \cap M' = \emptyset$, $M \subseteq M'$, or $M' \subseteq M$; a graph that only admits trivial modules is called \emph{prime}. Every non-singleton graph can be uniquely partitioned into a set of maximal strong modules $\mathcal{P} = \{M_1, \ldots, M_\ell\}$, with $\ell \geq 2$, called \emph{modular partition}. Two modules $M$ and $M'$ are said to be \emph{adjacent} if there exist $u \in M, v \in M'$ with $uv \in E$. In this case every vertex of $M$ is adjacent to every vertex of $M'$. By recursively partitioning the graphs $G[M_i]$ in this way, until every module is a single vertex, one obtains the \emph{modular decomposition} of $G$.

  The modular decomposition of $G$ can be represented as a rooted tree, where the root corresponds to $G$ and child nodes correspond to the graphs induced by the modules of their parents' modular partition. There are three possibilities for internal nodes of the modular decomposition: the quotient graph is a prime graph, clique, or independent set. Correspondingly, we call the node a \emph{prime node}, \emph{series node}, or \emph{parallel node}. The \emph{modular-width} of $G$, denoted $\mw(G)$, is the largest number of children of a prime node in the modular decomposition, but at least 2. The \emph{modular-depth} of $G$, denoted $\md(G)$, is the depth of the modular decomposition tree.

  Inspired by Coudert et al.~\cite{coudertcliquewidth}, we analyze the length of shortest augmenting paths with respect to the modular-width. Coudert et al.\ obtain a length of $\Theta(\mw(G))$ by replacing each module with a matching. The phase framework does not perform this replacement step, which causes us to be unable to control the augmenting path length with only the modular-width. As seen in \cref{lb_trivially_perfect}, there is a family of cographs, i.e., graphs with modular-width 2, that has unbounded shortest augmenting path length. It is still possible to bound the shortest augmenting path length in terms of an exponential function depending on the modular-width and the modular-depth. Note that this is only interesting due to the exclusion of parallel and series nodes in the definition of modular-width. Otherwise, this bound would be implied by a simple bound on the number of vertices.

  We will distinguish between two types of edges on alternating paths and bound them separately. Let $\mathcal{P} = \{M_1, \ldots, M_k\}$ be the modular partition of $G$. We say that an edge $e = vw \in E(G)$ is \emph{internal} if there is some $i$ so that $v,w \in M_i$, if such $i$ does not exist then $e$ is \emph{external}.

  \begin{lem}
    \label{module_border}
    Let $N$ be a matching in $G$ and let $\mathcal{P} = \{M_1, \ldots, M_k\}$ be its modular partition. Every $N$-alternating path can be replaced by an $N$-alternating path that uses at most 8 edges of $\delta(M_i)$ for all $i$.
  \end{lem}

  \begin{proof}
    Let $P = v_1 \ldots v_\ell$ be an $N$-alternating path and let $M = M_i$ for some $i$. We say that an edge $e = v_{i-1}v_i$ of $P$ \emph{enters} $M$ if $v_{i-1} \notin M$ and $v_{i} \in M$.

    First, we show that $P$ can have at most 2 red edges entering $M$. If
    \begin{align*}
      v_{i_1 - 1}v_{i_1}, \enskip \ldots, \enskip v_{i_k - 1}v_{i_k}, \quad i_1 < i_2 < \cdots < i_k, \quad \text{where } k \geq 3,
    \end{align*}
    are the red edges of $P$ that enter $M$, then $P$ can be replaced by
    \begin{align*}
      P' = P_{\left[v_1, v_{i_1}\right]} P_{\left[v_{i_k - 1}, v_\ell\right]},
    \end{align*}
    since $v_{i_1} v_{i_k - 1}$ exists and must be blue. Hence, at most 2 red edges entering $M$ remain.
    This can introduce additional blue edges that enter $M$, these will be bounded in the next step.

    In a similar manner, we show that $P$ can be replaced by a path $P'$ that has at most two blue edges entering $M$. Let
    \begin{align*}
      v_{i_1 - 1}v_{i_1}, \enskip \ldots, \enskip v_{i_k - 1}v_{i_k}, \quad i_1 < i_2 < \cdots < i_k, \quad \text{where } k \geq 3,
    \end{align*}
    be the blue edges of $P$ that enter $M$, then $P$ can be replaced by
    \begin{align*}
      P' = P_{\left[v_1, v_{i_1 - 1}\right]} P_{\left[v_{i_k}, v_{i_\ell}\right]}, \enskip \text{ or } P'' = P_{\left[v_1, v_{i_1 - 1}\right]}P_{\left[v_{i_k-1}, v_\ell\right]}.
    \end{align*}
    The edges $v_{i_1 - 1} v_{i_k}$ and $v_{i_1 - 1} v_{i_{k - 1}}$ exist due to $M$ being a module and at least one of them must be blue. If $v_{i_1 - 1} v_{i_k}$ is red, then we use $P''$ as replacement path.
    Hence, we have reduced the number of blue edges entering $M$ to at most 2.

    By considering $P$ in the reverse direction, this also bounds the number of edges of $P$ \emph{leaving} $M$, i.e., $e = v_{i-1} v_i$ with $v_{i-1} \in M$ and $v_i \notin M$. Observe that the replacements for entering edges could not have increased the number of leaving edges. In conclusion, we can replace $P$ in such a way that the resulting path uses at most 8 edges of $\delta(M)$.

    Notice that performing these replacements for $M_i$ does not increase the number of entering or leaving edges for any other $M_j$ with $j \neq i$. Hence, by iterating through all $M_i$ and performing the replacements
    we obtain the bound of 8 for each $\delta(M_i)$.
  \end{proof}

  \begin{lem}
    \label{series_node_external}
    Let $G$ be a series node and $N$ be a matching in $G$. Every $N$-alternating path in $G$ can be replaced so that it contains at most 4 blue external edges and at most 6 red external edges.
  \end{lem}

  \begin{proof}
    Let $P$ be an $N$-alternating path and let $\mathcal{P} = \{M_1, \ldots, M_k\}$ be the modular partition of $G$.
    Note that all $M_i$ are pairwise adjacent since $G$ is a series node.
    Fixing an orientation of $P$, an external edge $e = vw$ has type $(s,t)$ if $v \in M_s$ and $w \in M_t$.
    To ensure that the edges we use as shortcuts have the correct parity we discard the first and last edge of $P$,
    hence our upper bounds must be increased by 2.

    Let $b_1, \ldots, b_\ell$ be the remaining blue external edges on $P$ in order.
    If $i < j$ and $b_i = v_i w_i$ is of type $(s,t)$ and $b_j = v_j w_j$ of type $(s',t')$ and $s \neq t'$, then $P_{[v_i, w_j]}$ can be replaced by the
    direct edge $v_i w_j$. The edge $v_i w_j$ exists because $v_i$ and $w_j$ lie in different, but adjacent, modules. We also know that $v_i w_j$ has to be
    blue as $b_i$ is preceeded by a red edge and $b_j$ is succeeded by a red edge. Observe that this replacement removes the blue external edges
    $b_i, b_{i+1}, \ldots, b_j$ on $P$ and introduces one new blue external edge. As long as $\ell \geq 3$, such a situation must occur by the pigeonhole principle
    and we can decrease $\ell$ by at least one.

    Let $r_1, \ldots, r_\ell$ be the remaining red external edges on $P$ in order.
    If $i < j - 1$ and $r_i = v_i w_i$ is of type $(s,t)$ and $r_j = v_j w_j$ of type $(s',t')$ and $t \neq s'$, then $P_{[w_i, v_j]}$ can be replaced by the
    direct edge $w_i v_j$. The edge $v_j w_i$ exists and is blue similar to the previous case. In contrast to the previous case, this replacement does not
    remove the edges $r_i$ and $r_j$, but removes at least one red external edge between $r_i$ and $r_j$. As long as $\ell \geq 5$, such a situation occurs and
    we can decrease $\ell$ by at least one.

    When applying the replacements for the red external edges we introduce further blue external edges;
    to obtain simultaneous bounds on both, we first perform the replacements for the red external edges and afterwards for the blue external edges.
  \end{proof}

  \begin{thm}
    There exists $c \geq 1$ such that every graph $G$ is $(c \mw(G))^{\md(G)}$-replaceable.
  \end{thm}

  \begin{proof}
    We choose $c = 21$, but we did not make an effort to optimize this constant.

    We prove this statement by induction on the modular-depth $\md(G)$. If $\md(G) = 1$, then $G$ must be a prime graph and hence $n \leq \mw(G)$, so $G$ is trivially $\mw(G)$-replaceable.

    For $\md(G) \geq 2$, let $\mathcal{P} = \{M_1, \ldots, M_k\}$ be the modular partition of $G$. The graphs $G[M_i]$, $1 \leq i \leq k,$ have modular-depth at most $\md(G) - 1$.

    If $G$ is a parallel node, then each module corresponds to a connected component and hence $G$ is trivially $(c \mw(G))^{\md(G) - 1}$-replaceable by induction.

    If $G$ is a series node, $N$ a matching and $P$ an $N$-alternating path, then we invoke \cref{series_node_external} on $P$. After this replacement $P$ has at most 10 external edges and induces at most $11$ subpaths inside the modules $M_1, \ldots, M_k$. We can bound the length of these subpaths by induction and obtain a replacement path of length at most
    \begin{align*}
      10 + 11(c \mw(G))^{\md(G) - 1} \leq 21 (c \mw(G))^{\md(G) - 1} \leq (c \mw(G))^{\md(G)}.
    \end{align*}

    If $G$ is a prime node, $N$ a matching and $P$ an $N$-alternating path, then $k \leq \mw(G)$. Due to \cref{module_border}, we can replace $P$ in such a way that it induces at most 8 subpaths inside a module $M_i$, for each module $M_i$. By induction, these subpaths can be replaced so that each of them has length at most $(c \mw(G))^{\md(G) - 1}$. Every edge of an alternating path either belongs to some $\delta(M_i)$ or is inside a module $M_i$. In total, we can bound the length of the replacement path by
    \begin{align*}
      k (8 + 5 (c \mw(G))^{\md(G) - 1}) & \leq 8 \mw(G) + \mw(G) (c \mw(G))^{\md(G) - 1} \\
					& \leq 9 \mw(G) (c \mw(G))^{\md(G) - 1} \\
					& \leq (c \mw(G))^{\md(G)}. \qedhere
    \end{align*}
  \end{proof}

  By \cref{replaceable_phases}, we obtain the following time bounds.

  \begin{cor}
    Let $\class_{\mw, \md}$ be the class of graphs that have modular-width at most $\mw$ and modular-depth at most $\md$.
    There exists a constant $c \geq 1$ such that every algorithm employing the phase framework requires at most $\Oh(\sqrt{\dist_{\class_{\mw,\md}}(G)}(c\mw)^{\md})$ phases.
    In particular, \textsc{Maximum Matching} can be solved in time $\Oh(\sqrt{\dist_{\class_{\mw,\md}}(G)}(c\mw)^{\md} m)$.
  \end{cor}


\section{Lower bounds on the number of phases}\label{section:lowerbounds}

In this section, we show that several restrictive graph classes do not admit results such as those obtained in the previous section. To this end, we show that the assumptions made about algorithms that follow the phase framework still allow a worst case of $\Omega(\sqrt{n})$ phases. In other words, further assumptions about the behavior of such an algorithm are necessary to avoid these lower bounds and to again get adaptive running times.

\subsection{Paths and forests}

  \begin{figure}[h]
    \includegraphics[width = \textwidth]{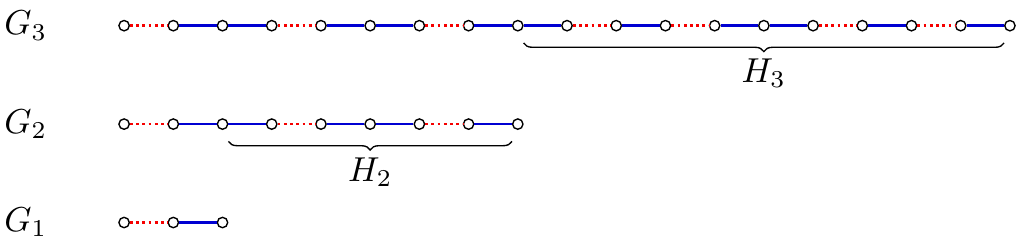}
    \caption{Lower bound construction for paths.}
    \label{fig:lower_bound_paths}
  \end{figure}

  \begin{lem}
    \label{lower_bound_paths}
    An algorithm employing the phase framework may choose a sequence of augmentations resulting in at least $\Omega(\sqrt{n})$ phases on paths.
  \end{lem}

  \begin{proof}
    We will concatenate increasingly long paths where we take every second edge into the matching in such a way that the matching on the newly added path
    is maximal but not maximum, so that every one of the concatenated paths requires a separate phase.
    More formally, on $P_{2i} = v_1 v_2 \ldots v_{2i}$ we use the matching $\{v_{2k}v_{2k+1} \, : \, k \in [i-1]\}$. This matching is one edge away from
    the maximum and the only augmenting path has length $2i - 1$.

    Let $H_i$ be obtained by concatenating two copies of $P_{2i}$ through identification of two endpoints (obtaining a path on $4i-1$ vertices). The matching on $H_i$ can be augmented once by an augmenting path of length $2i-1$. There are two augmenting paths on $H_i$, we
    always choose to augment along the copy of $P_{2i}$ that is attached to the previously constructed path, i.e., the left copy in \cref{fig:lower_bound_paths}. Using two copies of
    $P_{2i}$ in $H_i$ ensures that every $H_i$ requires at least one augmentation. This is not the case if we simply concatenate $P_2, P_4, \ldots, P_{2k}$.

    We can now define our desired paths. Let $G_1 = H_1 = P_3$ and in this exceptional case we assume that the left edge of $P_3$ is matched.
    Furthermore, let $G_{i+1}$ be obtained from $G_i$ by concatenating $H_{i+1}$ at the right.
    The number of vertices of $G_{\left\lceil \sqrt{n} \right\rceil}$ can be bounded by
    \begin{align*}
      \sum_{i = 1}^{\left\lceil \sqrt{n} \right\rceil} 2|V(P_{2i})| \leq 4 (\left\lceil \sqrt{n} \right\rceil)^2 \in \Oh(n).
    \end{align*}

    Now, we argue that our choice of augmentations leads to $\Omega(\sqrt{n})$ phases for $G_{\left\lceil \sqrt{n} \right\rceil}$. In the first phase we choose to find the maximal matching
    that we associated with each $H_i$. Now, observe that $G_{\left\lceil \sqrt{n} \right\rceil}$ contains augmenting paths of lengths $3, 5, \ldots, 2\left\lceil \sqrt{n} \right\rceil - 1$. In phase $i$, $i \geq 2$ we augment along the left copy of $P_{2i}$ in $H_i$. As this does not affect the augmenting paths in the other $H_i$, we need $\left\lceil \sqrt{n} \right\rceil$ phases in total.
  \end{proof}

Using the parameter values on paths, we obtain the following parameterized lower bounds.

  \begin{cor}
    An algorithm employing the phase framework may choose a sequence of augmentations resulting in $\Omega(\sqrt{\alpha(G)})$, $\Omega(\sqrt{\tau(G)})$, $\Omega(\sqrt{\dist_{\plex{1}}(G)})$ or $\Omega(\sqrt{\nd(G)})$ phases, where $\nd(G)$ is the neighborhood diversity of $G$.
  \end{cor}


\subsection{Cographs}
  \label{lb_trivially_perfect}

  Mertzios et al.~\cite{mnn} devised a \textsc{Maximum Matching} algorithm parameterized by vertex deletion distance to cocomparability graphs and there are several \textsc{Maximum Matching}
  algorithms parameterized by modular-width \cite{coudertcliquewidth, kratschnelles}. In the interest of showing that these results cannot be replicated or improved by our approach, we give a
  lower bound result for cographs, i.e., the graphs of modular-width 2 and a subclass of cocomparability graphs.

  \begin{figure}
    \centering
    \begin{subfigure}[b]{.45\textwidth}
      \centering
      \includegraphics[scale = 1.2]{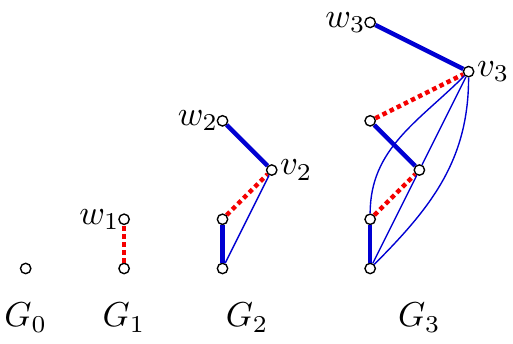}
      \caption{Increasingly long augmenting paths.}
    \end{subfigure}
    \begin{subfigure}[b]{.45\textwidth}
      \centering
      \includegraphics[scale = 1.2]{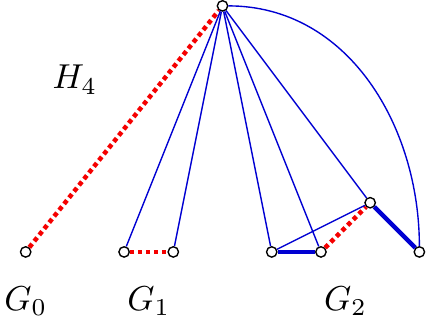}
      \caption{Construction of $H_4$.}
    \end{subfigure}
    \caption{Cographs with increasingly long augmenting paths, the red dotted edges are matched and the blue edges are unmatched.
      The thickened edges denote the chosen augmenting paths.}
    \label{fig:lower_bound_cographs}
  \end{figure}

  \begin{dfn}
    A graph is a \emph{cograph} if it can be constructed from the following operations:
    \begin{itemize}
    \item $K_1$ is a cograph,
    \item the disjoint union $G \cup H$ of two cographs $G$ and $H$ is a cograph,
    \item the join $G \times H$ of two cographs $G$ and $H$ is a cograph, where $V(G \times H) = V(G) \cup V(H)$ and $E(G \times H) = E(G) \cup E(H) \cup \{vw \, : \, v \in V(G), w \in V(H)\}.$
    \end{itemize}
  \end{dfn}

  \begin{thm}
    \label{lower_bound_cographs}
    There is a family of cographs such that an algorithm employing the phase framework may choose a sequence of augmentations resulting in $\Omega(\sqrt{n})$ phases
    on this family.
  \end{thm}

  \begin{proof}
    We will construct appropriate cographs using the cograph operations. Let $G_0 = K_1$ and $G_1 = K_1 \times K_1 = P_2$ and for $i \geq 1$ define $G_{i+1} = (G_i \cup K_1) \times K_1$
    as auxiliary graphs. Given $n$, we construct
    \begin{align*}
      H_n = \left( \bigcup_{i = 0}^{\left\lceil \sqrt{n} \right\rceil} G_i \right) \times G_0.
    \end{align*}

    Observe that $G_i$ has at most $2i + 1$ vertices, therefore $H_n$ has $\Oh(n)$ vertices. We will now describe a sequence of augmentations in $H_n$ that
    requires $\Omega(\sqrt{n})$ phases.

    In each $G_i$, $i \geq 2$, there is exactly one vertex of degree one, which we call $w_i$.
    For $G_1$, we fix an arbitrary vertex that is called $w_1$. Furthermore, the vertex that is joined last in the construction of $G_i$, $i \geq 2$, is referred to as $v_i$, see \cref{fig:lower_bound_cographs}.
    First, we describe which maximal matchings to associate with the graphs $G_i$. In $G_0$ there is no edge to choose, in $G_1$ we choose the only possible edge and in $G_2$
    we take the edge $v_2w_1$. Inductively, for $G_{i + 1}$, $i \geq 2,$ we use the maximal matching of $G_i$ and add the edge $v_{i+1}w_i$. With $H_n$ we associate the union of these matchings
    and add the edge between the two copies of $G_0$ to the matching. In this way we obtain the matching that is supposed to be found in the first phase.

    We now argue that $G_i$, $i \geq 2$, has exactly one shortest augmenting path of length $2i - 1$. This is true for $i = 2$; the other cases follow by induction.
    Let $P_i$ be the shortest augmenting path in $G_i$ starting at $w_i$, then $P_{i + 1} = w_{i+1} v_{i+1} P_i$ is an augmenting path of length $2i + 1$ in $G_{i + 1}$.
    There are no further augmenting paths as there are only two exposed vertices in $G_{i + 1}$ and one of them is $w_{i + 1}$ with degree one; starting at
    $w_{i + 1}$, we must first take the edge $w_{i + 1}v_{i + 1}$ and then the matched edge $v_{i + 1}w_i$, hence we must take the path $P_{i + 1}$ by induction.

    Let $v$ denote the vertex in $G_0$ that is joined last in the construction of $H_n$. The construction of $H_n$ does not create any additional augmenting paths as every maximal alternating path that passes through $v$ must have one matched endpoint, namely the vertex in the other copy of $G_0$. In phase $i$ we augment the shortest augmenting path of length $2i - 1$ in the copy of $G_i$. By repeating the previous argument, these augmentations
    cannot introduce any new augmenting paths in the later phases. Hence, we require $\left \lceil \sqrt{n} \right \rceil$ phases for this sequence of augmentations, thereby
    proving the claimed lower bound.
  \end{proof}

  The graphs in the previous proof are also $C_4$-free, therefore these graphs are not only cographs but also \emph{trivially perfect graphs}.

  
\subsection{Lower bound for bipartite chain graphs}
\label{sec_chain}
  Similar to \cref{lb_trivially_perfect}, Mertzios et al.\ gave a \textsc{Maximum Matching} algorithm for bipartite graphs parameterized by the vertex deletion distance to
  chain graphs \cite{mnn}. We show that such a result does not hold for arbitrary algorithms that follow the phase framework.

  \begin{dfn}[\cite{yannakakis1982complexity}]
    A graph $G$ is a \emph{bipartite chain graph} if $G$ is bipartite with bipartition $V = A \cup B$ and there is an ordering of the vertices of $A = \{a_1, \ldots, a_k\}$
    and an ordering of the vertices of $B = \{b_1, \ldots, b_{\ell} \}$ such that $N(a_i) \subseteq N(a_{i+1})$ for all $i \in [k-1]$ and $N(b_{i + 1}) \subseteq N(b_i)$ for all
    $i \in [\ell-1]$.
  \end{dfn}

  \begin{figure}
    \centering
    \includegraphics[scale = 1.3]{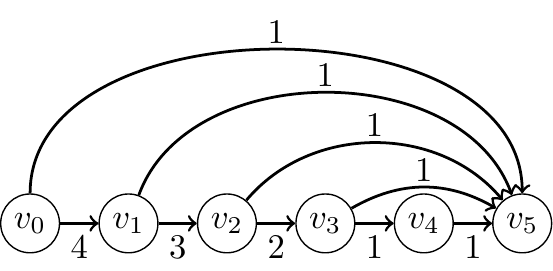}
    \caption{The directed graph $D_5$.}
    \label{fig:lower_bound_chains_intuition}
  \end{figure}

  The rough idea of the lower bound construction for chain graphs is to encode the directed graph $D_k$, depicted in \cref{fig:lower_bound_chains_intuition}, as a chain graph, where $D_k$ is given by
  \begin{align*}
    V(D_k) & = \{v_i \, : \, i = 0, \ldots, k \} \\
    A(D_k) & = \{v_i v_{i + 1} \, : \, i = 0, \ldots, k - 2\} \cup \{ v_i v_k \, : \, i = 0, \ldots, k - 1 \}
  \end{align*}
  with multiplicities, i.e., how many parallel copies there are of each edge,
  \begin{align*}
    w(v_i v_{i+1}) = k - 1 - i \text{ and } w(v_i v_k) = 1.
  \end{align*}
  Any $v_0, v_k$-path $P$ in $D_k$ of length $\ell$ will correspond to an augmenting path of length $2 \ell + 1$ in the chain graph. Notice that there is exactly one $v_0, v_k$-path
  in $D_k$ for each length $\ell = 1, \ldots, k$ and by making use of the multiplicities these are all edge-disjoint. The consequence in the corresponding chain graph is that each phase will
  only find one augmenting path and hence we must perform at least $k$ phases. The construction of the chain graph will ensure that $k \in \Theta(\sqrt{n})$ by replacing each $v_i$ with multiple vertices
  and we choose a specific initial matching, i.e., in phase 1, to model the behavior of $D_k$. Hence we obtain the desired lower bound for chain graphs. We will not make the relation between $D_k$ and the chain graph precise, this paragraph only serves as intuition.

  \begin{figure}
    \centering
    \includegraphics[scale = 1.3]{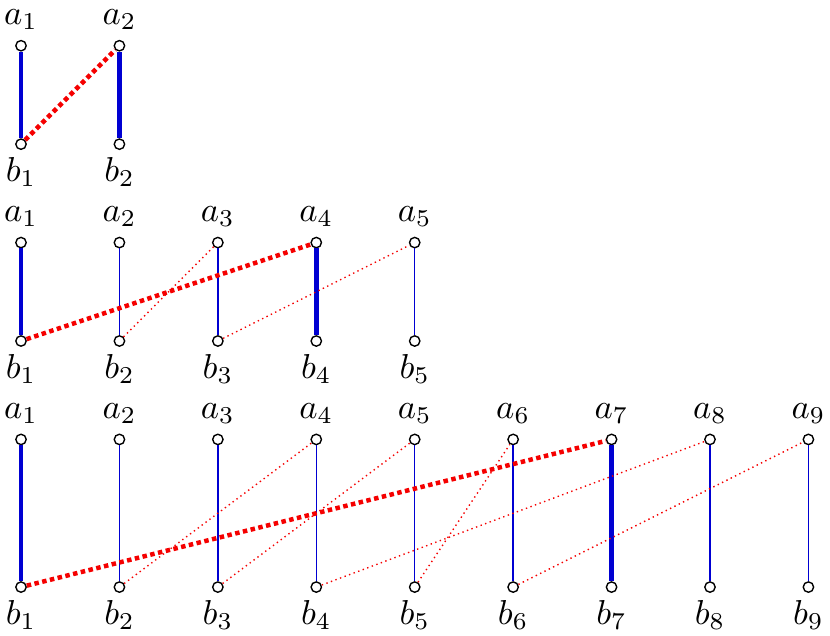}
    \caption{The graphs $G_2$, $G_3$ and $G_4$ with the initial matching $M_0$.}
    \label{fig:lower_bound_chains}
  \end{figure}

  \begin{figure}
    \centering
    \includegraphics[scale = 1.3]{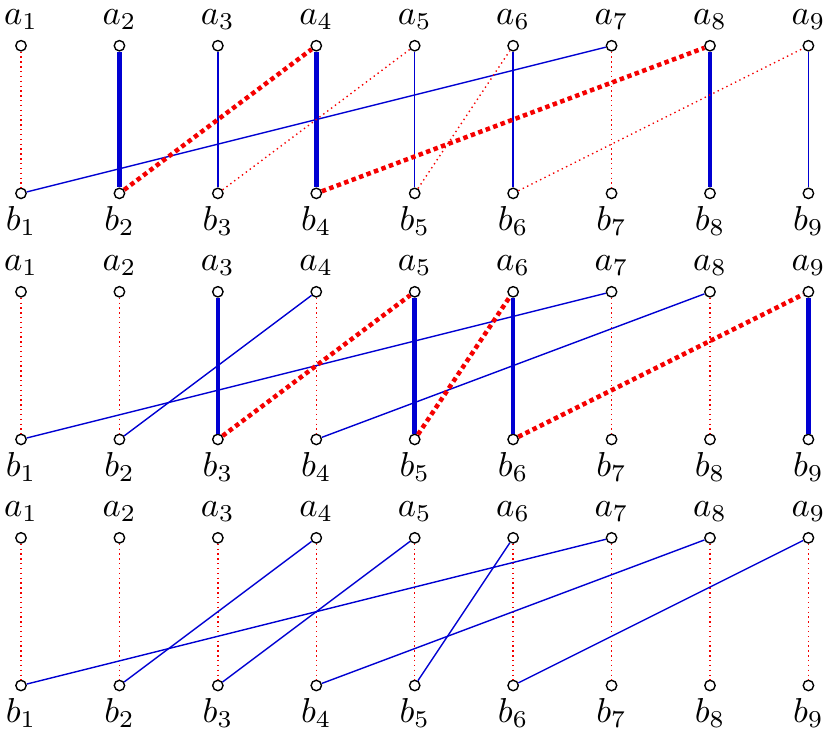}
    \caption{The graph $G_4$ with the matchings $M_1$, $M_2$ and $M_3$.}
    \label{fig:lower_bound_chains_2}
  \end{figure}

  \begin{thm}
    \label{lower_bound_chains}
    There is a family of bipartite chain graphs such that an algorithm employing the phase framework may choose a sequence of augmentations resulting in at least $\Theta(\sqrt{n})$ phases
    on this family.
  \end{thm}

  \begin{proof}
    Fix $k \in \mathbb{N}$ and set $n = \sum_{i = 2}^k i = \frac{k(k+1)}{2} - 1$. We define the graph $G_k$ with $2n$ vertices by
    \begin{align*}
      V(G_k) & = A \cup B, \text{ where } A = \{a_i \, : \, i \in [n]\} \text{ and } B = \{b_i \, : \, i \in [n]\}, \\
      E(G_k) & = \{a_ib_j \, : \, i \geq j\}.
    \end{align*}
    As $N(a_i) = \{b_1, b_2, \ldots, b_i\}$ and $N(b_i) = \{a_i, a_{i+1}, \ldots, a_n\}$ for all $i \in [n]$, the graph $G_k$ must be a bipartite chain graph, because
    the inclusions $N(a_i) \subsetneq N(a_{i+1})$ and $N(b_{i+1}) \subsetneq N(b_i)$ hold for all $i \in [n-1]$.

    We partition the index set $[n]$ into $k + 1$ intervals as follows
    \begin{align*}
      \idx_j = \left[1 + \sum_{i = 0}^{j-1} (k - i), \sum_{i = 0}^j (k - i)\right] \cap \mathbb{N} \text{ for } j = 0, \ldots, k - 1
    \end{align*}
    and
    \begin{align*}
      \idx_k = \left[n - k + 1, n\right] \cap \mathbb{N}.
    \end{align*}
    Observe that $|\idx_j| = k - j$ for $j = 0, \ldots, k - 1$ and $|\idx_k| = k$. Furthermore, we define $i^*_j = \min \idx_j = 1 + \sum_{i = 0}^{j - 1} (k-i)$ for $j = 0, \ldots, k - 1$.
    The vertex set $\{a_i \, : \, i \in \idx_j\} \cup \{b_i \, : \, i \in \idx_j\}$ roughly corresponds to the vertex $v_j$ of the graph $D_k$.
    We can now define the initial matching
    \begin{align*}
      M_0 = \left\{ b_{i^*_j} a_{n-k+1+j} \, : \, 0 \leq j < k \right\} \cup \bigcup_{j = 0}^{k-2} \left\{ b_s a_{s+k-1-j} \, : \, s \in \idx_j \setminus \{i^*_j\} \right\}.
    \end{align*}
    The graphs $G_2, G_3$ and $G_4$ are depicted in \cref{fig:lower_bound_chains} with their initial matching.
    The $M_0$-exposed vertices are $\{a_i \, : \, i \in \idx_0 \} \cup \{b_i \, : \, i \in \idx_k \}$, hence $M_0$ must be maximal as $i < j$ for all $i \in \idx_0$ and $j \in \idx_k$.

    As $G_k$ is bipartite, we can assume that every augmenting path starts somewhere in $A$, ends somewhere in $B$ and is oriented from its endpoint in $A$ to its endpoint in $B$.

    We will now define $k - 1$ vertex-disjoint $M_0$-augmenting paths $P_\ell$, $\ell = 1, \ldots, k-1$. The path $P_\ell$ starts at vertex $a_\ell$ and every time we are at some $a_i$, $i \in [n]$, we take the blue edge $a_ib_i$ and every time we are at some $b_i$, $i \in [n]$, we take the incident red edge. Since no edge of the form $a_ib_i$ is red, this construction must result in a path. We have partitioned $[n]$ in such a way that if we take a red edge from a $b_i$ with $i \in \idx_j$, then we get to an $a_{i'}$ with $i' \in \idx_{j'}$ and $j' > j$. More precisely, if $i \neq i^*_j$ then $j' = j + 1$ and otherwise $j' = k$. Using this observation, we see that $P_\ell$ is an $M_0$-augmenting path of length $2\ell + 1$ that starts at $a_\ell$ and ends in $b_{n-k+\ell}$, because $P_\ell$ takes $2\ell - 1$ edges to reach $b_{i^*_{\ell - 1}}$ and then two further edges to reach $b_{n-k+\ell}$.

    Now, we can define further matchings $M_\ell = M_{\ell-1} \triangle E(P_\ell)$ for $\ell = 1, \ldots, k - 1$. Due to the $P_\ell$ being vertex-disjoint we see that every $P_\ell$ is also an $M_{\ell-1}$-augmenting path, thus the matchings $M_\ell$ are well-defined. The resulting sequence of matchings for $G_4$ is depicted in \cref{fig:lower_bound_chains_2} and they are explicitly given by
    \begin{align*}
      M_\ell 	= & \bigcup_{j = 0}^{k-2} \left\{ b_s a_{s+k-1-j} \, : \, s \in \idx_j \setminus \left[i^*_j, i^*_j + \max(0, \ell - 1 - j)\right] \right\}\\
		  & \cup \left\{ b_{i^*_j} a_{n-k+1+j} \, : \, \ell \leq j < k \right\} .
    \end{align*}

    We claim that $M_0, M_1, \ldots, M_{k-1}$ can be computed by the phases of an algorithm employing the phase framework. For this we must show that there is no $M_\ell$-augmenting path of length at most $2\ell + 1$ for all $\ell = 0, \ldots, k-1$. To see this, notice that the last red edge taken by any $M_\ell$-augmenting path is one of the edges in $\left\{ b_{i^*_j} a_{n-k+1+j} \, : \, \ell \leq j < k \right\}$. Hence, we first must go to some $a_i$ with $i \geq i^*_\ell$. For $j < \ell$ going from an $a_s$ with $s \in \idx_j$ to an $a_t$ with $t \in \bigcup_{i = j + 1}^k \idx_i$ requires at least one red edge and using exactly one red edge we can only get to an $a_t$ with $t \in \idx_{j + 1}$. Hence, getting from $\idx_0$ to $\idx_j$, $\ell \leq j < k$, requires at least $\ell$ red edges in any $M_\ell$-augmenting path and one further red edge to get to $\idx_k$, thus any $M_\ell$-augmenting path has length at least $2(\ell + 1) + 1 = 2\ell + 3 > 2\ell + 1$.

    In conclusion, this sequence of matchings leads to $k$ phases and due to $|V(G_k)| \in \Theta(k^2)$ we obtain the desired lower bound.
  \end{proof}


\section{Conclusion}\label{section:conclusion}

We have conducted an adaptive analysis that applies to all algorithms for \textsc{Maximum Matching} that follow the phase framework of Hopcroft and Karp~\cite{hopcroftkarp}, such as the algorithms due to Micali and Vazirani~\cite{micalivazirani}, Blum~\cite{blum1990new}, and Goldberg and Karzanov~\cite{goldbergkarzanov}. The main take-away message of our paper is that these algorithms not only obtain the best known time $\Oh(\sqrt{n}m)$ for solving \textsc{Maximum Matching} but that they are also (obliviously) adaptive to beneficial structure. That is, they run in linear time on several graph classes and they run in time $\Oh(\sqrt{k}m)$ for graphs that are $k$ vertex deletions away from any of several classes; before, most bounds were $\Omega(km)$. Arguably, such adaptive algorithms are the best possible result for dealing with unknown beneficial structure because they are never worse than the general bound, in this case taking $\Omega(\sqrt{n}m)$ when $k=\Theta(n)$, and smoothly interpolate to linear time on well-structured instances. Moreover, in the present case, they unify several special cases and remove the need to find exact or approximate beneficial structure.

We complemented our findings by proving that the phase framework alone still allows taking $\Omega(\sqrt{n})$ phases, and, hence, total time $\Omega(\sqrt{n}m)$, even on restrictive classes like paths, trivially perfect graphs, and bipartite chain graphs (and their superclasses), despite the existence of (dedicated) linear-time algorithms. Of course, all of these cases are easy to handle but it raises the question whether there are simple further properties to demand of a phase-based algorithm so that it is provably adaptive to larger classes such as cocomparability or bounded treewidth graphs? In the same vein, it would be interesting whether time $\Oh(\sqrt{k}m)$ is possible relative to feedback vertex number $k$, i.e., relative to deletion distance to a forest.

More generally, with the large interest in ``FPT in P'' (or efficient parameterized algorithms), it seems interesting what other fundamental problems admit adaptive algorithms that interpolate between, say, linear time and the best general bound. Are there other cases where a proven algorithmic paradigm, like the path packing phases of Hopcroft and Karp~\cite{hopcroftkarp}, also obliviously yields the best known running times relative to beneficial input structure?

\newpage
\bibliography{paper_lipics_final_full}

\begin{thebibliography}{10}

\bibitem{abboud2016approximation}
Amir Abboud, Virginia~Vassilevska Williams, and Joshua Wang.
\newblock Approximation and fixed parameter subquadratic algorithms for radius
  and diameter in sparse graphs.
\newblock In {\em Proceedings of the twenty-seventh annual ACM-SIAM symposium
  on Discrete Algorithms}, pages 377--391. SIAM, 2016.

\bibitem{AbboudWY15}
Amir Abboud, Virginia~Vassilevska Williams, and Huacheng Yu.
\newblock Matching triangles and basing hardness on an extremely popular
  conjecture.
\newblock In Rocco~A. Servedio and Ronitt Rubinfeld, editors, {\em Proceedings
  of the Forty-Seventh Annual {ACM} on Symposium on Theory of Computing, {STOC}
  2015, Portland, OR, USA, June 14-17, 2015}, pages 41--50. {ACM}, 2015.
\newblock URL: \url{https://doi.org/10.1145/2746539.2746594}, \href
  {http://dx.doi.org/10.1145/2746539.2746594}
  {\path{doi:10.1145/2746539.2746594}}.

\bibitem{BarbayFN12}
J{\'{e}}r{\'{e}}my Barbay, Johannes Fischer, and Gonzalo Navarro.
\newblock Lrm-trees: Compressed indices, adaptive sorting, and compressed
  permutations.
\newblock {\em Theor. Comput. Sci.}, 459:26--41, 2012.
\newblock URL: \url{https://doi.org/10.1016/j.tcs.2012.08.010}, \href
  {http://dx.doi.org/10.1016/j.tcs.2012.08.010}
  {\path{doi:10.1016/j.tcs.2012.08.010}}.

\bibitem{BarbayN13}
J{\'{e}}r{\'{e}}my Barbay and Gonzalo Navarro.
\newblock On compressing permutations and adaptive sorting.
\newblock {\em Theor. Comput. Sci.}, 513:109--123, 2013.
\newblock URL: \url{https://doi.org/10.1016/j.tcs.2013.10.019}, \href
  {http://dx.doi.org/10.1016/j.tcs.2013.10.019}
  {\path{doi:10.1016/j.tcs.2013.10.019}}.

\bibitem{bast2006matching}
Holger Bast, Kurt Mehlhorn, Guido Schafer, and Hisao Tamaki.
\newblock Matching algorithms are fast in sparse random graphs.
\newblock {\em Theory of Computing Systems}, 39(1):3--14, 2006.

\bibitem{bentert2017parameterized}
Matthias Bentert, Till Fluschnik, Andr{\'e} Nichterlein, and Rolf Niedermeier.
\newblock Parameterized aspects of triangle enumeration.
\newblock In {\em International Symposium on Fundamentals of Computation
  Theory}, pages 96--110. Springer, 2017.

\bibitem{berge1957two}
Claude Berge.
\newblock Two theorems in graph theory.
\newblock {\em Proceedings of the National Academy of Sciences},
  43(9):842--844, 1957.

\bibitem{blum1990new}
Norbert Blum.
\newblock A new approach to maximum matching in general graphs.
\newblock In {\em International Colloquium on Automata, Languages, and
  Programming}, pages 586--597. Springer, 1990.

\bibitem{Bringmann14}
Karl Bringmann.
\newblock Why walking the dog takes time: Frechet distance has no strongly
  subquadratic algorithms unless {SETH} fails.
\newblock In {\em 55th {IEEE} Annual Symposium on Foundations of Computer
  Science, {FOCS} 2014, Philadelphia, PA, USA, October 18-21, 2014}, pages
  661--670. {IEEE} Computer Society, 2014.
\newblock URL: \url{https://doi.org/10.1109/FOCS.2014.76}, \href
  {http://dx.doi.org/10.1109/FOCS.2014.76} {\path{doi:10.1109/FOCS.2014.76}}.

\bibitem{chang1996algorithms}
Maw-Shang Chang.
\newblock Algorithms for maximum matching and minimum fill-in on chordal
  bipartite graphs.
\newblock In {\em International Symposium on Algorithms and Computation}, pages
  146--155. Springer, 1996.

\bibitem{coudertcliquewidth}
David Coudert, Guillaume Ducoffe, and Alexandru Popa.
\newblock Fully polynomial {FPT} algorithms for some classes of bounded
  clique-width graphs.
\newblock In {\em Proceedings of the Twenty-Ninth Annual ACM-SIAM Symposium on
  Discrete Algorithms}, pages 2765--2784. Society for Industrial and Applied
  Mathematics, 2018.

\bibitem{cournier1994new}
Alain Cournier and Michel Habib.
\newblock A new linear algorithm for modular decomposition.
\newblock In {\em Colloquium on Trees in Algebra and Programming}, pages
  68--84. Springer, 1994.

\bibitem{dahlhaus1998matching}
Elias Dahlhaus and Marek Karpinski.
\newblock Matching and multidimensional matching in chordal and strongly
  chordal graphs.
\newblock {\em Discrete Applied Mathematics}, 84(1-3):79--91, 1998.

\bibitem{DisserK17}
Yann Disser and Stefan Kratsch.
\newblock Robust and adaptive search.
\newblock In Heribert Vollmer and Brigitte Vall{\'{e}}e, editors, {\em 34th
  Symposium on Theoretical Aspects of Computer Science, {STACS} 2017, March
  8-11, 2017, Hannover, Germany}, volume~66 of {\em LIPIcs}, pages 26:1--26:14.
  Schloss Dagstuhl - Leibniz-Zentrum fuer Informatik, 2017.
\newblock URL: \url{https://doi.org/10.4230/LIPIcs.STACS.2017.26}, \href
  {http://dx.doi.org/10.4230/LIPIcs.STACS.2017.26}
  {\path{doi:10.4230/LIPIcs.STACS.2017.26}}.

\bibitem{dragan1997greedy}
Feodor~F. Dragan.
\newblock On greedy matching ordering and greedy matchable graphs.
\newblock In {\em International Workshop on Graph-Theoretic Concepts in
  Computer Science}, pages 184--198. Springer, 1997.

\bibitem{ducoffepopasplitwidth}
Guillaume Ducoffe and Alexandru Popa.
\newblock A quasi linear-time b-matching algorithm on distance-hereditary
  graphs and bounded split-width graphs.
\newblock {\em arXiv preprint arXiv:1804.09393}, 2018.

\bibitem{ducoffepopapruned}
Guillaume Ducoffe and Alexandru Popa.
\newblock {The use of a pruned modular decomposition for Maximum Matching
  algorithms on some graph classes}.
\newblock In {\em {29th International Symposium on Algorithms and Computation
  (ISAAC 2018)}}, 29th International Symposium on Algorithms and Computation
  (ISAAC 2018), Jiaoxi, Yilan County, Taiwan, December 2018.
\newblock URL: \url{https://hal.archives-ouvertes.fr/hal-01955985}, \href
  {http://dx.doi.org/10.4230/LIPIcs.ISAAC.2018.144}
  {\path{doi:10.4230/LIPIcs.ISAAC.2018.144}}.

\bibitem{edmonds1965paths}
Jack Edmonds.
\newblock Paths, trees, and flowers.
\newblock {\em Canadian Journal of Mathematics}, 17:449–467, 1965.
\newblock \href {http://dx.doi.org/10.4153/CJM-1965-045-4}
  {\path{doi:10.4153/CJM-1965-045-4}}.

\bibitem{Estivill-CastroW92}
Vladimir Estivill{-}Castro and Derick Wood.
\newblock A survey of adaptive sorting algorithms.
\newblock {\em {ACM} Comput. Surv.}, 24(4):441--476, 1992.
\newblock URL: \url{https://doi.org/10.1145/146370.146381}, \href
  {http://dx.doi.org/10.1145/146370.146381} {\path{doi:10.1145/146370.146381}}.

\bibitem{fialadistancelabeling}
Jir{\'{\i}} Fiala, Tomas Gavenciak, Dusan Knop, Martin Kouteck{\'{y}}, and Jan
  Kratochv{\'{\i}}l.
\newblock Parameterized complexity of distance labeling and uniform channel
  assignment problems.
\newblock {\em Discrete Applied Mathematics}, 248:46--55, 2018.
\newblock URL: \url{https://doi.org/10.1016/j.dam.2017.02.010}, \href
  {http://dx.doi.org/10.1016/j.dam.2017.02.010}
  {\path{doi:10.1016/j.dam.2017.02.010}}.

\bibitem{fluschnik2017can}
Till Fluschnik, Christian Komusiewicz, George~B Mertzios, Andr{\'e}
  Nichterlein, Rolf Niedermeier, and Nimrod Talmon.
\newblock When can graph hyperbolicity be computed in linear time?
\newblock In {\em Workshop on Algorithms and Data Structures}, pages 397--408.
  Springer, 2017.

\bibitem{Fomin2018}
Fedor~V. Fomin, Mathieu Liedloff, Pedro Montealegre, and Ioan Todinca.
\newblock Algorithms parameterized by vertex cover and modular width, through
  potential maximal cliques.
\newblock {\em Algorithmica}, 80(4):1146--1169, Apr 2018.
\newblock URL: \url{https://doi.org/10.1007/s00453-017-0297-1}, \href
  {http://dx.doi.org/10.1007/s00453-017-0297-1}
  {\path{doi:10.1007/s00453-017-0297-1}}.

\bibitem{fomin2018fully}
Fedor~V. Fomin, Daniel Lokshtanov, Saket Saurabh, Micha{\l} Pilipczuk, and
  Marcin Wrochna.
\newblock Fully polynomial-time parameterized computations for graphs and
  matrices of low treewidth.
\newblock {\em ACM Transactions on Algorithms (TALG)}, 14(3):34, 2018.

\bibitem{fouquet1999bipartite}
Jean-Luc Fouquet, Vassilis Giakoumakis, and Jean-Marie Vanherpe.
\newblock Bipartite graphs totally decomposable by canonical decomposition.
\newblock {\em International Journal of Foundations of Computer Science},
  10(04):513--533, 1999.

\bibitem{fouquet1997n}
Jean-Luc Fouquet, Igor Parfenoff, and Henri Thuillier.
\newblock An {$\mathcal{O}$}(n) time algorithm for maximum matching in
  {$P_4$}-tidy graphs.
\newblock {\em Information Processing Letters}, 62(6):281--287, 1997.

\bibitem{gabow1991faster}
Harold~N. Gabow and Robert~E. Tarjan.
\newblock Faster scaling algorithms for general graph matching problems.
\newblock {\em Journal of the ACM (JACM)}, 38(4):815--853, 1991.

\bibitem{gallai1967transitiv}
Tibor Gallai.
\newblock Transitiv orientierbare graphen.
\newblock {\em Acta Mathematica Hungarica}, 18(1-2):25--66, 1967.

\bibitem{ganian2012using}
Robert Ganian.
\newblock Using neighborhood diversity to solve hard problems.
\newblock {\em arXiv preprint arXiv:1201.3091}, 2012.

\bibitem{ganian2013expanding}
Robert Ganian and Jan Obdr{\v{z}}{\'a}lek.
\newblock Expanding the expressive power of monadic second-order logic on
  restricted graph classes.
\newblock In {\em International Workshop on Combinatorial Algorithms}, pages
  164--177. Springer, 2013.

\bibitem{gardi2003efficient}
Fr{\'e}d{\'e}ric Gardi.
\newblock Efficient algorithms for disjoint matchings among intervals and
  related problems.
\newblock In {\em Discrete Mathematics and Theoretical Computer Science}, pages
  168--180. Springer, 2003.

\bibitem{giannopoulou2017polynomial}
Archontia~C. Giannopoulou, George~B. Mertzios, and Rolf Niedermeier.
\newblock Polynomial fixed-parameter algorithms: A case study for longest path
  on interval graphs.
\newblock {\em Theoretical Computer Science}, 689:67--95, 2017.

\bibitem{glover1967maximum}
Fred Glover.
\newblock Maximum matching in a convex bipartite graph.
\newblock {\em Naval Research Logistics Quarterly}, 14(3):313--316, 1967.

\bibitem{goldbergkarzanov}
Andrew~V. Goldberg and Alexander~V. Karzanov.
\newblock Maximum skew-symmetric flows and matchings.
\newblock {\em Mathematical Programming}, 100(3):537--568, 2004.

\bibitem{guo2009more}
Jiong Guo, Christian Komusiewicz, Rolf Niedermeier, and Johannes Uhlmann.
\newblock A more relaxed model for graph-based data clustering: s-plex editing.
\newblock In {\em International Conference on Algorithmic Applications in
  Management}, pages 226--239. Springer, 2009.

\bibitem{hopcroftkarp}
John~E. Hopcroft and Richard~M. Karp.
\newblock An {$n^{5/2}$} algorithm for maximum matchings in bipartite graphs.
\newblock {\em SIAM Journal on computing}, 2(4):225--231, 1973.

\bibitem{husfeldtgraphdistances}
Thore Husfeldt.
\newblock {Computing Graph Distances Parameterized by Treewidth and Diameter}.
\newblock In Jiong Guo and Danny Hermelin, editors, {\em 11th International
  Symposium on Parameterized and Exact Computation (IPEC 2016)}, volume~63 of
  {\em Leibniz International Proceedings in Informatics (LIPIcs)}, pages
  16:1--16:11, Dagstuhl, Germany, 2017. Schloss Dagstuhl--Leibniz-Zentrum fuer
  Informatik.
\newblock URL: \url{http://drops.dagstuhl.de/opus/volltexte/2017/6947}, \href
  {http://dx.doi.org/10.4230/LIPIcs.IPEC.2016.16}
  {\path{doi:10.4230/LIPIcs.IPEC.2016.16}}.

\bibitem{iwatatreedepth}
Yoichi Iwata, Tomoaki Ogasawara, and Naoto Ohsaka.
\newblock On the power of tree-depth for fully polynomial {FPT} algorithms.
\newblock In {\em LIPIcs-Leibniz International Proceedings in Informatics},
  volume~96. Schloss Dagstuhl-Leibniz-Zentrum fuer Informatik, 2018.

\bibitem{kellerhals2018parameterized}
Leon Kellerhals.
\newblock Parameterized algorithms for network flows.
\newblock Master's thesis, TU Berlin, 2018.
\newblock URL:
  \url{https://fpt.akt.tu-berlin.de/publications/thesis/MA-leon-kellerhals.pdf}.

\bibitem{knop2017simplified}
Du{\v{s}}an Knop, Martin Kouteck{\`y}, Tom{\'a}{\v{s}} Masa{\v{r}}{\'\i}k, and
  Tom{\'a}{\v{s}} Toufar.
\newblock Simplified algorithmic metatheorems beyond {MSO}: Treewidth and
  neighborhood diversity.
\newblock In {\em International Workshop on Graph-Theoretic Concepts in
  Computer Science}, pages 344--357. Springer, 2017.

\bibitem{kratschnelles}
Stefan Kratsch and Florian Nelles.
\newblock {Efficient and Adaptive Parameterized Algorithms on Modular
  Decompositions}.
\newblock In Yossi Azar, Hannah Bast, and Grzegorz Herman, editors, {\em 26th
  Annual European Symposium on Algorithms (ESA 2018)}, volume 112 of {\em
  Leibniz International Proceedings in Informatics (LIPIcs)}, pages
  55:1--55:15, Dagstuhl, Germany, 2018. Schloss Dagstuhl--Leibniz-Zentrum fuer
  Informatik.
\newblock URL: \url{http://drops.dagstuhl.de/opus/volltexte/2018/9518}, \href
  {http://dx.doi.org/10.4230/LIPIcs.ESA.2018.55}
  {\path{doi:10.4230/LIPIcs.ESA.2018.55}}.

\bibitem{lampis2012algorithmic}
Michael Lampis.
\newblock Algorithmic meta-theorems for restrictions of treewidth.
\newblock {\em Algorithmica}, 64(1):19--37, 2012.

\bibitem{liang1993finding}
Y.~Daniel Liang and Chongkye Rhee.
\newblock Finding a maximum matching in a circular-arc graph.
\newblock {\em Information Processing Letters}, 45(4):185--190, 1993.

\bibitem{madry2013navigating}
Aleksander Madry.
\newblock Navigating central path with electrical flows: From flows to
  matchings, and back.
\newblock In {\em Foundations of Computer Science (FOCS), 2013 IEEE 54th Annual
  Symposium on}, pages 253--262. IEEE, 2013.

\bibitem{mcconnell1999modular}
Ross~M. McConnell and Jeremy~P. Spinrad.
\newblock Modular decomposition and transitive orientation.
\newblock {\em Discrete Mathematics}, 201(1-3):189--241, 1999.

\bibitem{mnn}
George~B. Mertzios, Andr{\'e} Nichterlein, and Rolf Niedermeier.
\newblock Fine-grained algorithm design for matching.
\newblock Technical report, Technical Report, 2016.

\bibitem{mnnKernel}
George~B. Mertzios, Andr{\'e} Nichterlein, and Rolf Niedermeier.
\newblock {The Power of Linear-Time Data Reduction for Maximum Matching}.
\newblock In Kim~G. Larsen, Hans~L. Bodlaender, and Jean-Francois Raskin,
  editors, {\em 42nd International Symposium on Mathematical Foundations of
  Computer Science (MFCS 2017)}, volume~83 of {\em Leibniz International
  Proceedings in Informatics (LIPIcs)}, pages 46:1--46:14, Dagstuhl, Germany,
  2017. Schloss Dagstuhl--Leibniz-Zentrum fuer Informatik.
\newblock URL: \url{http://drops.dagstuhl.de/opus/volltexte/2017/8116}, \href
  {http://dx.doi.org/10.4230/LIPIcs.MFCS.2017.46}
  {\path{doi:10.4230/LIPIcs.MFCS.2017.46}}.

\bibitem{mertzios2018linear}
George~B. Mertzios, Andr{\'e} Nichterlein, and Rolf Niedermeier.
\newblock A linear-time algorithm for maximum-cardinality matching on
  cocomparability graphs.
\newblock {\em SIAM Journal on Discrete Mathematics}, 32(4):2820--2835, 2018.

\bibitem{micalivazirani}
Silvio Micali and Vijay~V. Vazirani.
\newblock An {$O(\sqrt{|V|} |E|)$} algorithm for finding maximum matching in
  general graphs.
\newblock In {\em 21st Annual Symposium on Foundations of Computer Science,
  1980}, pages 17--27. IEEE, 1980.

\bibitem{nevsetril2012sparsity}
J.~Ne{\v{s}}et{\v{r}}il and P.~Ossona de~Mendez.
\newblock {\em Sparsity ({Graphs}, {Structures}, and {Algorithms}), volume 28
  of {Algorithms} and {Combinatorics}}.
\newblock Springer, 2012.

\bibitem{PatrascuW10}
Mihai Patrascu and Ryan Williams.
\newblock On the possibility of faster {SAT} algorithms.
\newblock In Moses Charikar, editor, {\em Proceedings of the Twenty-First
  Annual {ACM-SIAM} Symposium on Discrete Algorithms, {SODA} 2010, Austin,
  Texas, USA, January 17-19, 2010}, pages 1065--1075. {SIAM}, 2010.
\newblock URL: \url{https://doi.org/10.1137/1.9781611973075.86}, \href
  {http://dx.doi.org/10.1137/1.9781611973075.86}
  {\path{doi:10.1137/1.9781611973075.86}}.

\bibitem{seidman1978graph}
Stephen~B. Seidman and Brian~L. Foster.
\newblock A graph-theoretic generalization of the clique concept.
\newblock {\em Journal of Mathematical sociology}, 6(1):139--154, 1978.

\bibitem{steiner1996linear}
George Steiner and Julian~S Yeomans.
\newblock A linear time algorithm for maximum matchings in convex, bipartite
  graphs.
\newblock {\em Computers \& Mathematics with Applications}, 31(12):91--96,
  1996.

\bibitem{tedder2008simpler}
Marc Tedder, Derek Corneil, Michel Habib, and Christophe Paul.
\newblock Simpler linear-time modular decomposition via recursive factorizing
  permutations.
\newblock In {\em International Colloquium on Automata, Languages, and
  Programming}, pages 634--645. Springer, 2008.

\bibitem{van2012approximation}
Ren{\'e} Van~Bevern, Hannes Moser, and Rolf Niedermeier.
\newblock Approximation and tidying—a problem kernel for s-plex cluster
  vertex deletion.
\newblock {\em Algorithmica}, 62(3-4):930--950, 2012.

\bibitem{yannakakis1982complexity}
Mihalis Yannakakis.
\newblock The complexity of the partial order dimension problem.
\newblock {\em SIAM Journal on Algebraic Discrete Methods}, 3(3):351--358,
  1982.

\bibitem{yu1993n}
Ming-Shing Yu and Cheng-Hsing Yang.
\newblock An {$\mathcal{O}$}(n) time algorithm for maximum matching on
  cographs.
\newblock {\em Information Processing Letters}, 47(2):89--93, 1993.

\bibitem{yuster2013maximum}
Raphael Yuster.
\newblock Maximum matching in regular and almost regular graphs.
\newblock {\em Algorithmica}, 66(1):87--92, 2013.

\bibitem{yuster2007minors}
Raphael Yuster and Uri Zwick.
\newblock Maximum matching in graphs with an excluded minor.
\newblock In {\em Proceedings of the eighteenth annual ACM-SIAM symposium on
  Discrete algorithms}, pages 108--117. Society for Industrial and Applied
  Mathematics, 2007.

\end{thebibliography}

\end{document}